\documentclass{vldb}

\makeatletter %
\let\proof\@undefined
\let\endproof\@undefined
\makeatother %

\usepackage{xcolor}
\usepackage{graphicx}
\usepackage{amssymb}
\usepackage{amsthm}
\usepackage{amsmath}
\usepackage{amsfonts}
\usepackage[ruled,noend,linesnumbered]{algorithm2e}
\usepackage{balance}
\usepackage{textcomp}
\usepackage{textcomp}
\usepackage{url}
\usepackage{enumerate}

\newlength{\figwidths}
\newlength{\expwidths}
\newlength{\expwidthd}
\newtheorem{example}{Example}%[section]
\newtheorem{theorem}{Theorem}%[section]
\newtheorem{lemma}[theorem]{Lemma}%[section]
\newtheorem{proposition}[theorem]{Proposition}%[section]

\DeclareMathOperator*{\argmin}{arg\,min}

\makeatother

\begin{document}

\title{Time Series Data Cleaning: From Anomaly Detection to Anomaly Repairing (Technical Report)}

\numberofauthors{4}
\author{
\alignauthor
Aoqian Zhang\\ \vspace{0.1em}
       \affaddr{School of Software, \\Tsinghua University}\\
       \textrm{zaq13@mails.tsinghua.edu.cn}
\alignauthor
Shaoxu Song\\ \vspace{0.1em}
       \affaddr{School of Software, \\Tsinghua University}\\
       \textrm{sxsong@tsinghua.edu.cn}
\alignauthor
Jianmin Wang\\ \vspace{0.1em}
       \affaddr{School of Software, \\Tsinghua University}\\
       \textrm{jimwang@tsinghua.edu.cn}
\vspace{0.5em}\and
\alignauthor
Philip S. Yu\\ \vspace{0.1em}
       \affaddr{University of Illinois at Chicago}\\
       \affaddr{Institute for Data Science, Tsinghua University}\\
       \textrm{psyu@cs.uic.edu}
}

\maketitle

\begin{abstract}
Errors are prevalent in time series data, such as GPS trajectories or sensor readings. 
Existing methods focus more on anomaly detection but not on repairing the detected anomalies. 
By simply filtering out the dirty data via anomaly detection, applications could still be unreliable over the incomplete time series. 
Instead of simply discarding anomalies, we propose to (iteratively) repair them in time series data, 
by creatively bonding the beauty of temporal nature in anomaly detection with 
the widely considered minimum change principle in data repairing.
Our major contributions include: 
(1) a novel framework of iterative minimum repairing (IMR) over time series data, 
(2) explicit analysis on convergence of the proposed iterative minimum repairing, 
and 
(3) efficient estimation of parameters in each iteration. 
Remarkably, with incremental computation, we reduce the complexity of parameter estimation from $O(n)$ to $O(1)$.
Experiments on real datasets demonstrate the superiority of our proposal compared to the state-of-the-art approaches.
In particular, we show that (the proposed) repairing indeed improves the time series classification application.
\end{abstract}

%-------------------------------------------------------------------------
\section{Introduction}
\label{sect:introduction}

Time series data are often found with dirty or imprecise values, such as GPS trajectories, sensor reading sequences \cite{DBLP:conf/vldb/JefferyGF06}, or even stock prices \cite{DBLP:journals/pvldb/LiDLMS12}.
For example,  
the price of SALVEPAR (SY) is misused as the price of SYBASE (SY), both of which share the same notation (SY) in some sources. 
It is different from the interesting anomaly that actually happens in real life, e.g., 
the temperatures sudden change from 20C to 10C in one day when cold air rushes in.
To distinguish such cases, we propose to employ some labeled truth of dirty observations.
(See more detailed examples on dirty data and their labeled truth in Example \ref{example:motivation}.)

\subsection{Motivation on Anomaly Repairing}
\label{sect-introduction-anomaly}

Applications, such as pattern mining \cite{DBLP:conf/kdd/Morchen06} 
or classification \cite{DBLP:journals/kais/XingPY12}, 
built upon the dirty time series data are obviously not reliable. 
Anomaly detection over time series  is often applied to filter out the dirty data
(see \cite{gupta2014outlier} for a comprehensive and structured overview of anomaly detection techniques).
That is, the detected anomaly data points are simply discarded as useless noises. 
Unfortunately, with a large number of consecutive data points eliminated,
the applications could be barely performed over the rather incomplete time series. 

Recent study \cite{DBLP:conf/kdd/SongLZ15} shows that repairing dirty values could improve clustering over spatial data.
For time series data, we argue that repairing the anomaly can also improve the applications such as time series classification \cite{DBLP:journals/kais/XingPY12}. 
A repair close to the truth helps greatly the applications.

\subsection{Potential Methods for Repairing}
\label{sect-introduction-existing}

A straightforward idea is to directly interpret the predication values in anomaly detection, 
e.g., by AR \cite{box1994time,hill2010anomaly,yamanishi2002unifying} or ARX \cite{box1994time,park2005outlier}, 
as repairs (see details in Section \ref{sect:model}). 
A data point is considered as anomaly if its (truth) predication significantly differs from (noisy) observation. 
Unfortunately, noisy/erroneous data are often close to the truth in practice, 
under the intuition that human or systems always try to minimize their mistakes, 
e.g., misspellings (John Smith vs. Jhon Smith), 
typos (555-8145 vs. 555-8195) 
as illustrated in \cite{DBLP:conf/sigmod/BohannonFFR05},
rounding off (76,821,000 vs. 76M) or unit error (76M vs. 76B) 
as shown in \cite{DBLP:journals/pvldb/LiDLMS12}.
Owing to such disagreement, the repairing performance of directly applying anomaly detection techniques is poor, as illustrated in both Example \ref{example:motivation} below and experiments in Section \ref{sect:experiment}.

On the other hand, constraint-based repairing SCREEN \cite{DBLP:conf/sigmod/SongZWY15}, strictly following the minimum change principle in data repairing \cite{DBLP:conf/sigmod/BohannonFFR05}, 
heavily relies on a proper constraint of speeds on value changes. 
The repairing is performed based on two consecutive points, i.e., considering only one historical point, and thus 
does not sense the temporal nature of errors. 
As shown in the following Example \ref{example:motivation}, 
the speed constraint-based SCREEN is effective in repairing  spike errors, 
but can hardly handle a sequence of consecutive dirty points.

In short, the anomaly detection method does not expect the minimized mistakes in practice, 
whereas the constraint-based repairing is not effective in addressing the temporal nature of errors.

\subsection{Intuition of Our Proposal}
\label{sect-introduction-intuition}

Since completely automatic data repairing might not work well in  repairing time series data (such as SCREEN \cite{DBLP:conf/sigmod/SongZWY15} observed in our experiments in Figure \ref{exp:ild-consize}),
enlightened by the idea of utilizing master data (a single repository of high-quality data) in data repairing \cite{DBLP:journals/pvldb/FanLMTY10}, 
we employ some labeled truth of dirty observations to advance the repair.
The truth can be obtained either by manual labeling or automatically by more reliable sources.  
For instance, accurate locations are manually marked in the map by user check-in activities (and utilized to repair the imprecise GPS readings). 
Moreover, periodical automatic labeling may take place in certain scenarios, e.g., precise equipments report accurate air quality data  (as labeled values) in a relatively long sensing period, 
while crowd and participatory sensing generates unreliable observations in a constant manner \cite{DBLP:conf/kdd/ZhengLH13}.

Being aware of both the error nature in anomaly detection and the minimum change principle in data repairing, 
we propose \emph{iterative minimum repairing} (IMR). 
The philosophy behind the proposed iterative minimum repairing is that the high confidence repairs in the former iterations could help the latter repairing.
Specifically, IMR minimally changes one point a time to obtain the most confident repair only, referring to the minimum change principle in data repairing that human or systems always try to minimize their mistakes. 
The high confidence repairs, 
together with the labeled truth of error points, 
are utilized to learn and enhance the temporal nature of errors in anomaly detection, and thus generate more accurate repair candidates in the latter iterations.

\begin{figure}[t]
\centering
\includegraphics[width=\figwidths]{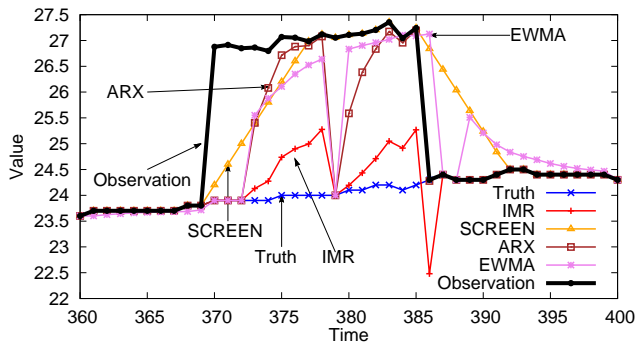}
\caption{An example segment of sensor readings}
\label{fig:simple-example}
\end{figure}

\begin{example}
\label{example:motivation}
Figure \ref{fig:simple-example} presents an example segment of sensor readings, denoted by black line. 
Suppose that sensor errors occur in the period from time point 370 to 385, 
where the observations  are shifted from the truth, 
e.g., owing to granularity mismatch or unit error.
To repair the errors, the truth of several observations are labeled, including time points 
$\{370, 371, 372, 379, 387\}$.

Existing speed constraint-based cleaning (SCREEN) \cite{DBLP:conf/sigmod/SongZWY15} could not effectively repair such continuous errors in a period 
(which is indeed also observed in Example 1 in \cite{DBLP:conf/sigmod/SongZWY15}). 
The reason is that speed constraints, restricting the amount of value changing relative to time difference, can detect sharp deviations such as from time point 369 to 370, but not continuous errors, e.g., in 383 and 384.  
The exponentially weighted moving average (EWMA) \cite{hellerstein2008quantitative} algorithm also hard to find a proper way to clean the trace. These two methods have similar repair trace.

By considering the predication values in anomaly detection as repairs (see more details in Section \ref{sect:model}), the result of ARX based repairing is also reported. 
ARX, considering the errors between truth and observations, shows better repair results than EWMA and SCREEN methods. 
 
Finally, our proposed IMR approach, with both error predication and minimum change considerations, obtains repairs  closest to the truth.
\end{example}

The iterative minimum repairing leads to new challenges: 
(1) whether the repairing process converges; and 
(2) how to efficiently/incrementally update the parameter of the temporal model over the repaired data after each iteration.
Both issues in anomaly repairing are not considered in the  anomaly detection studies.%

\paragraph*{Contributions}

Our major contributions in this paper are summarized as follows.

(1) We formalize the anomaly repairing problem, given a time series with some points having labeled truth, in Section \ref{sect:model}. 
The adaption of existing anomaly detection techniques (such as AR and ARX) is introduced for anomaly repairing.

(2) We devise an iterative minimum-change-aware repairing algorithm IMR, in Section \ref{sect:framework}. 
Remarkably, we illustrate that the ARX-based approach (in Section \ref{sect:model}) is indeed a special case of IMR with static parameter (Proposition \ref{proposition:fixed-phi-one-round}).

(3) We study the convergence of IMR in various scenarios, in Section \ref{sect:analysis}. 
In particular, the convergence is explicitly analyzed for the special case of IMR(1) with order $\mathit{p}=1$, 
which is sufficient to achieve high repair accuracy in practice (as shown in the experiments in Section \ref{sect:experiment}).
We prove that under certain inputs, the converged repair result could be directly calculated without iterative computing (Proposition~\ref{proposition:dynamic-ar1-one}).  

(4) We design efficient pruning and incremental computation for parameter estimation in each repair iteration, in Section \ref{sect:prune}. 
Rather than performing parameter estimation over all the $n$ points, 
matrices for parameter estimation could be pruned by simply removing rows with value 0 (Proposition \ref{proposition:matrix-reduce}). 
It is also remarkable that the incremental computation among different repair iterations (Proposition \ref{the:incremental}) could further reduce the complexity of parameter estimation from 
$O(n)$ to $O(1)$.

(5) Experiments on real datasets with both real and synthetic errors, in Section \ref{sect:experiment}, demonstrate that 
our IMR method shows significantly better repair performance than the state-of-the-art approaches, 
including the aforesaid anomaly detection and constraint-based repairing. 
Table \ref{table:notations} 
lists the notations  frequently used  in this paper.

\begin{table}[t]
 \caption{Notations}
 \label{table:notations}
 \centering
 \begin{tabular}{rl}
 \noalign{\smallskip}
 \hline\noalign{\smallskip} Symbol & Description \\ \noalign{\smallskip}
 \hline\noalign{\smallskip}
 $\mathit{x}$ & observation sequence of $n$ data points \\ \noalign{\smallskip}
 $\mathit{x}_i$ & value of $i$-th data point in $\mathit{x}$, a.k.a. $\mathit{x}[i]$ \\  \noalign{\smallskip}
 $\mathit{y}$ & truth-labeled/repaired sequence of $\mathit{x}$ \\ \noalign{\smallskip}
 $\mathit{z}$ & distance between $\mathit{x}$ and labeled/repaired $\mathit{y}$  \\ \noalign{\smallskip} 
\noalign{\smallskip}
 $\mathit{y}^{(k)}$ & sequence $\mathit{y}$ in the $k$-th iteration \\ \noalign{\smallskip}
 $\phi$ & parameter of AR($p$)/ARX($p$) with order $\mathit{p}$   \\ \noalign{\smallskip} 
 $\tau$ & predefined threshold of convergence \\ \noalign{\smallskip} 
 $\boldsymbol{\mathit{Z}}, \boldsymbol{\mathit{V}}$ & input matrices for parameter estimation \\ \noalign{\smallskip}
 $\boldsymbol{\mathit{A}}, \boldsymbol{\mathit{B}}$ & intermediate matrices for  parameter estimation \\ \noalign{\smallskip} \hline
 \end{tabular}
 \end{table}

%-------------------------------------------------------------------------
\section{Preliminaries}\label{sect:model}

This section first introduces
the problem of anomaly repairing.  
We then adapt the existing anomaly detection models for anomaly repairing, 
i.e., AR without considering labeled data and ARX supporting labeled data.

The major issues of this simple adaption are:
(1) Applying predications with significant difference to the observation as repairs 
contradicts to the minimum change principle in data repairing \cite{DBLP:conf/sigmod/BohannonFFR05}, 
as discussed in the introduction. 
(2) A static parameter ($\phi$ in Equations \ref{equation:AR} and \ref{equation:ARX} below) needs to be preset, e.g., estimated from the dirty data during the initialization.

\subsection{Problem Statement}
\label{sect-probelm-statement}

Consider a time series of $n$  observations, $\mathit{x} = \mathit{x}[1],\dots,\mathit{x}[n]$,  
where each $\mathit{x}[i]$ is the value of the $i$-th data point. 
For brevity, we write $\mathit{x}[i]$ as $\mathit{x}_i$.

Let $\mathit{y}$ denote the labeled/repaired sequence of $\mathit{x}$. 
Each $\mathit{y}_i$ is either the labeled truth or the repaired value of $\mathit{x}_i$.

Given a time series $\mathit{x}$ and a 
partially labeled subset $\mathit{y}$ of $\mathit{x}$,
the repairing problem is to determine  the repairs $\mathit{y}_i$ of $\mathit{x}_i$ that are not labeled in $\mathit{y}$.

\begin{example}[Observation $\mathit{x}$,  partially labeled $\mathit{y}$, and fully repaired $\mathit{y}$]
\label{example:input}
Consider $\mathit{x} = \{6, 10, 9.6, 8.3, 7.7, 5.4, 5.6, 5.9, 6.3, \allowbreak 6.8, 7.5, 8.5\}$ 
with twelve data points of observations
in Figure \ref{fig:example-dirty}, 
where shift (up) errors occur on four points $x_2,\dots,x_5$. 
Suppose that five points are labeled with truth, i.e., the partially labeled $\mathit{y}$.
By repairing (using the methods presented below), we propose to obtain a fully repaired $\mathit{y}$, e.g., 
$\mathit{y} = \{6, 5.6, 5.4, 5.2, 5.4, 5.4, 5.6, 5.9, 6.3, 6.8, 7.5, 8.5\}$ as shown in Figure \ref{fig:example-dirty}. 
In the repaired $\mathit{y}$, $\mathit{x}_4$ and $\mathit{x}_5$ are changed from 8.3 and 7.7 to 5.2 and 5.39, respectively. 
The labeled $\mathit{y}_2$ and $\mathit{y}_3$ will not be modified in the repair result.

Figure \ref{fig:example-dirty} also presents another repair $y'$ by the approach of connecting the dots with the labeled values, 
i.e., linear interpolation \cite{wiener1949extrapolation}. 
As shown in Figure \ref{fig:example-dirty} and also indicated in \cite{DBLP:conf/sigmod/SongZWY15}, the major issue of this (smoothing) approach is the serious damage of almost all the (unlabeled) data points, 
such as $y'_{7} \dots y'_{11}$, 
which are originally correct and should not be modified. 
In contrast, 
our proposed method repairs $y_{4}$ and $y_{5}$ while leaving $y_{7} \dots y_{11}$ unchanged. 
\end{example}

\subsection{AR Model}
\label{sect:ar}

Intuitively, anomaly detection techniques could be adapted to  anomaly repairing. 
For instance, we consider the AR (autoregression) model \cite{hill2010anomaly, yamanishi2002unifying} as follows: 
\begin{align}
\mathit{x}'_t &= c + \sum\limits_{i=1}^p\phi_i\mathit{x}_{t-i} + \epsilon_t 
\label{equation:AR}
\end{align}
where 
$\mathit{x}'_t$ is the prediction of $\mathit{x}_t$, 
$\mathit{p}$ is the order,  
$\phi_i$ is the parameter of the model, 
$c$ is a constant defined by 
$c = \mu(1-\sum\limits_{i=1}^p\phi_i)$, 
$\mu$ is mean value of the process and 
$\epsilon_t$ is white noise
(usually Gaussian white noise \cite{box1994time}, a normal random variable generated according to the Gaussian distribution with mean $\mu=0$ and variance $\sigma^2$; in other words, $c=0$).

If $\mathit{x}'_t$ significantly differs from the original observation $\mathit{x}_t$, having 
$|\mathit{x}'_t-\mathit{x}_t|>\tau$ where $\tau$ is a predefined threshold, 
this predication is accepted $\mathit{x}_t=\mathit{x}'_t$, a.k.a.\ a repair. 
The intuition behind is that a farther distance indicates the higher probability of being an outlier.
The threshold $\tau$ can be decided by observing the statistical distribution of distances between $\mathit{x}_{\mathit{t}}' $ and $\mathit{x}_
{\mathit{t}}$,  using the prediction interval \cite{hill2010anomaly,han2011data}.

The AR-based repairing procedure is thus: 
(1) replace $\mathit{x}_t$ by $\mathit{y}_t$ if it is labeled,
(2) learn parameter $\phi$ of AR($\mathit{p}$) from $\mathit{x}$, and 
(3) fill all unlabeled $\mathit{y}_t$ by AR($\mathit{p}$) over $\mathit{x}$, having 
\begin{align}\label{equation:AR-repair}
\mathit{y}_t = 
\begin{cases}
\mathit{x}'_t & \quad \textrm{if } \mathit{y}_t \text{ is unlabeled and } |\mathit{x}'_t-\mathit{x}_t|>\tau \\
\mathit{x}_t  & \quad \textrm{otherwise } 
\end{cases} 
\end{align}

\begin{figure}[t]
\centering
\includegraphics[width=\figwidths]{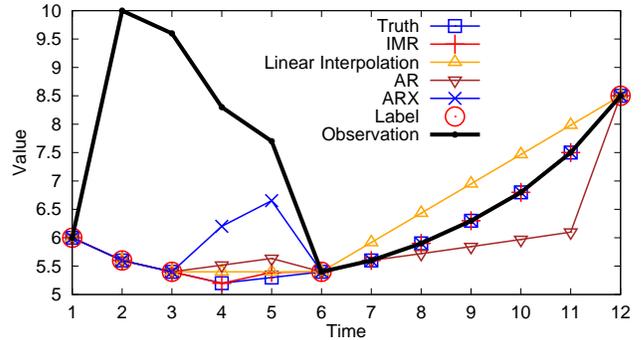}
\caption{Example of observations and repairs}
\label{fig:example-dirty}
\end{figure}

\begin{example}[Example \ref{example:input} continued]
\label{example:AR}
Consider again $\mathit{x} = \{6, 10, 9.6, 8.3, 7.7, 5.4, 5.6, 5.9, 6.3, 6.8, 7.5, 8.5\}$
in Figure \ref{fig:example-dirty}. 
For simplicity, we use AR(1) with order $\mathit{p} = 1$ and $c$ = 0, i.e., $\mathit{x}_t' = \phi_1\mathit{x}_{t-1}$. 
By ordinary least square \cite{rao2009linear}, we estimate the parameter $\phi$ from $\mathit{x}$, having $\phi_1 = 1.022$. 
Let $\tau = 0.1$. 
$\mathit{y}_{1}$ is labeled with truth. 
Referring to Equation \ref{equation:AR-repair}, it outputs unchanged $\mathit{y}_{1}  = 6$.
Similarly for $\mathit{y}_{2}$ and $\mathit{y}_{3}$. 
We have $\mathit{x}_{4}' = \phi_1\mathit{x}_1 = 1.022 * 5.4 = 5.52$. 
Since $|5.52 - 8.3| = 2.78 > 0.1$, $\mathit{x}_{4}'$ is accepted as new $\mathit{x}_{4}$.

Similarly, we have $\mathit{x}_{5}' = \phi_1\mathit{x}_{4} = 1.022 * 5.52 = 5.64$. 
Referring to $|5.64 - 7.7| = 2.06 > 0.1$, the prediction is accepted.
So on and so forth, we obtain the final repaired result 
$\mathit{y} = \{6, 5.6, 5.4, 5.52, 5.64, 5.4, 5.6, 5.72, 5.84, 5.97, 6.10, 8.5\}$. 
Its RMS error is 0.51 (see Section \ref{sect:experiment} for RMS definition).
\end{example}

\subsection{ARX Model}
\label{sect:ARX}

In order to utilize the labeled $\mathit{y}$, we consider the ARX model (autoregressive model with exogenous inputs) \cite{park2005outlier} 
\begin{align}
\mathit{y}'_t = \mathit{x}_t + \sum\limits_{i=1}^p\phi_i(\mathit{y}_{t-i}-\mathit{x}_{t-i}) + \epsilon_t 
\label{equation:ARX}
\end{align}
where $\mathit{y}'_t$ is the possible repair of $\mathit{x}_t$,
and others are the same to the aforesaid AR model. 
As shown in Equation \ref{equation:ARX}, not only the preceding observations 
$\mathit{x}_{t-i}$ will affect the determination of $\mathit{y}'_t$, 
but also the previously labeled/repaired $\mathit{y}_{t-i}$.

The ARX-based repairing procedure is thus: 
(1) learn parameter $\phi$ of ARX($\mathit{p}$) from $\mathit{x}$ and partially labeled $\mathit{y}$, and 
(2) fill all unlabeled $\mathit{y}_t$ by ARX($\mathit{p}$), similar to Equation \ref{equation:AR-repair} by replacing $\mathit{x}'_t$ with $\mathit{y}'_t$.

\begin{example}[Example \ref{example:input} continued]
\label{example:ARX}
Consider  again $\mathit{x} = \{6, 10, 9.6, 8.3, 7.7, 5.4, 5.6, 5.9, 6.3, 6.8, 7.5, 8.5\}$ and the partially labeled $\mathit{y}$ over five points, in Figure \ref{fig:example-dirty}.
For simplicity, we also use ARX(1) with order $\mathit{p} = 1$, i.e., $\mathit{y}_t' = \mathit{x}_t + \phi_1(\mathit{y}_{t-1}-\mathit{x}_{t-1})$. 
Similar to AR in Example \ref{example:AR}, we estimate the parameter $\phi$ by ordinary least square \cite{rao2009linear}, having $\phi_1 = 0.5$. 
Let $\tau = 0.1$.
Again, the labeled $\mathit{y}_{3}=5.4$ is not modified. 
For the fourth point, we have $\mathit{y}_{4}' = 8.3 + 0.5*(5.4-9.6) = 6.2$. Since $|6.2 - 8.3| = 2.1 > 0.1$, 
we assign $\mathit{y}_{4} = \mathit{y}_{4}' = 6.2$.
Finally, the  repair result by ARX is $\mathit{y} = \{6, 5.6, 5.4, 6.20, 6.65, 5.4, 5.6, 5.9, 6.3, 6.8, 7.5, 8.5\}$ 
with RMS error 0.49, lower than that of AR in Example \ref{example:AR}.
\end{example}

We consider ARX model since it can capture the difference between the observed errors and labeled truths, 
while other methods such as AR and SCREEN ignore (cannot utilize) such differences.
By modeling such differences between errors and truths rather than original values, 
both ARX and our proposed IMR may not deal with Spike errors (i.e., with \# consecutive errors = 1 in Figure \ref{exp:ild-consize}) as good as SCREEN \cite{DBLP:conf/sigmod/SongZWY15}.
Nevertheless, with iterative repairing, 
IMR always shows significantly better results in addressing a large number of consecutive errors (see Figure \ref{exp:ild-consize} as well). 
ARX cannot address well such consecutive errors either, since it applies only significant changes which is contradict to the minimum change principle in data repairing as discussed in Section \ref{sect-introduction-existing} in the Introduction.

%-------------------------------------------------------------------------

\section{Repair Algorithm}\label{sect:framework}

Unlike the existing anomaly detection model that may over-change the data via one-pass repairing (as discussed in the introduction, illustrated in Figure \ref{fig:example-dirty}, and observed in experiments in Section \ref{sect:experiment}),
we propose to progressively repair the data, being aware of both the error nature in anomaly detection and the minimum change principle in data repairing, 
so that the high confidence repairs in the former iterations could help the repairing in the latter steps.%

\subsection{Iterative Repairing}
\label{sect:overview}

Let $\mathit{y}^{(k)}$ denote the sequence $\mathit{y}$ in the $k$-th iteration, 
where $\mathit{y}^{(0)}$ is the partially labeled time series in the input. 
Since $\mathit{y}^{(0)}$ is incomplete (partially labeled), 
to initialize, 
we assign $\mathit{y}_t^{(0)}=\mathit{x}_t$ if $\mathit{y}_t^{(0)}$ is not labeled. 
Recall that the labeled values should not be repaired, i.e., 
$\mathit{y}^{(k)}_t=\mathit{y}^{(0)}_t$ if $\mathit{y}^{(0)}_t$ is labeled. 
 
Algorithm \ref{algorithm:arx-iteration} presents the iterative minimum repairing procedure, IMR($\mathit{p}$),
whose inputs are the observation time series $\mathit{x}$ and partially labeled $\mathit{y}^{(0)}$. 
It outputs $\mathit{y}^{(k)}$ with all the labeled $\mathit{y}^{(0)}_t$ unchanged and unlabeled $\mathit{y}^{(0)}_t$ repaired.

The major steps include: 

\textbf{(S1) Parameter estimation}, in Line \ref{algorithm:repair-parameter}, 
learns the parameter of ARX($p$) in the $k$-th iteration, denoted by $\phi^{(k)}$, from $\mathit{x}$  and the current $\mathit{y}^{(k)}$.

\textbf{(S2) Repair candidate generation}, in Line \ref{algorithm:repair-candidate}, 
computes the possible repairs $\hat{\mathit{y}}^{(k)}$, according to ARX($\mathit{p}$) 
w.r.t.\ $\mathit{x}$, $\mathit{y}^{(k)}$ and $\phi^{(k)}$.

\textbf{(S3) Repair evaluation}, 
in Line \ref{algorithm:repair-evaluate}, 
determines one of the repairs to accept, 
$ \mathit{y}_t^{(k+1)} = \hat{\mathit{y}}_t^{(k)}$, 
referring to the minimum change principle in data repairing \cite{DBLP:conf/sigmod/BohannonFFR05}.

As shown in Line \ref{algorithm:repair-converge}, the procedure repeats, until the repair converges, 
e.g., having 
\begin{align}
\label{equation:convergence}
|\mathit{y}_j^{(k)} - \mathit{y}_j^{(k+1)}| \leq \tau , 
j = 1, \dots,  n.
\end{align}
where $\tau$ a threshold of convergence,
or a maximum number of iterations is reached.
Setting \emph{max-num-iterations} is a remedy to avoid waiting for convergence in practice (see Section \ref{sect:max-iteration-exp} for discussion and evaluation). 

\begin{algorithm}\label{algorithm:arx-iteration}
\caption{IMR($\mathit{p}$)}
 \KwIn{time series $\mathit{x}$ and partially labeled $\mathit{y}^{(0)}$}
 \KwOut{$\mathit{y}^{(k)}$ with all the labeled $\mathit{y}^{(0)}_i$ unchanged and unlabeled $\mathit{y}^{(0)}_j$ repaired}

\For{$k \leftarrow 0$ \KwTo max-num-iterations}{
 $\phi^{(k)} \leftarrow \mathsf{Estimate}(\mathit{x}, \mathit{y}^{(k)})$\; 
 \label{algorithm:repair-parameter}
 $\hat{\mathit{y}}^{(k)} \leftarrow \mathsf{Candidate}( \mathit{x}, \mathit{y}^{(k)}, \phi^{(k)})$\; \label{algorithm:repair-candidate}
 $\mathit{y}^{(k+1)} \leftarrow \mathsf{Evaluate}(\mathit{x}, \mathit{y}^{(k)}, \hat{\mathit{y}}^{(k)})$\;
 \label{algorithm:repair-evaluate}
 \If{$\mathsf{Converge}(\mathit{y}^{(k)}, \mathit{y}^{(k+1)})$\label{algorithm:repair-converge}}{
 \textbf{break}\;
 }
 $k\leftarrow k+1$\;
}
 \Return{$\mathit{y}^{(k)}$}
\end{algorithm} 

\begin{example}[Algorithm overview, Example \ref{example:input} continued]
\label{example:algorithm-overview}
Consider again $\mathit{x} = \{6, 10, 9.6, 8.3, 7.7, 5.4, 5.6, 5.9, 6.3, 6.8, \allowbreak 7.5, 8.5\}$
in Figure \ref{fig:example-dirty}. 
According to five labeled data points, 
we assign $\mathit{y}^{(0)} = \{6, 5.6, 5.4, 8.3, 7.7, 5.4, 5.6, 5.9, 6.3, 6.8, 7.5, 8.5\}$, 
where the unlabeled points are initialized by $\mathit{y}_t^{(0)} = \mathit{x}_t$,  
e.g., $\mathit{y}_{4}^{(0)} = \mathit{x}_{4} = 8.3$. 

In each iteration, the IMR algorithm 
(1) learns the parameter, e.g., $\phi_1^{(0)}=0.5$ for $\mathit{p}=1$;
(2) generates candidates for repairing, such as $\hat{\mathit{y}}^{(0)}  = \{-, -, -, 6.2, 7.7, -, 5.6, 5.9, 6.3, 6.8, \allowbreak 7.5, -\}$; and
(3) selects one repair to conduct, and form the new sequence, 
say $\mathit{y}^{(1)} = \{6, 5.6, 5.4, {6.2}, 7.7, 5.4, 5.6, 5.9, 6.3,\allowbreak 6.8, 7.5, 8.5\}$.

The procedure repeats until converging.
The final output is $\mathit{y}^{(7)} = \{6, 5.6, 5.4, 5.20, 5.39, 5.4, 5.6, 5.9, 6.3, 6.8, 7.5, 8.5\}$ with RMS error 0.03.
Details on each step are presented in the following examples.
\end{example}

\subsection{Parameter Estimation}
\label{sect:parameter-estimation}

The parameter estimation step S1 (in Line \ref{algorithm:repair-parameter} in Algorithm \ref{algorithm:arx-iteration}) 
estimates the parameter $\phi^{(k)}$ for ARX($\mathit{p}$), 
given $\mathit{x}, \mathit{y}^{(k)}$. 
Existing methods such as 
Ordinary Least Square \cite{rao2009linear}
or 
Yule-Walker Equations \cite{Cheng:2014aa}
can be directly employed. 
For instance, by Ordinary Least Square, 
we have 
\begin{align}
\label{equation:ols}
\phi^{(k)} &= ((\boldsymbol{\mathit{Z}}^{(k)})'\boldsymbol{\mathit{Z}}^{(k)})^{-1}(\boldsymbol{\mathit{Z}}^{(k)})'\boldsymbol{\mathit{V}}^{(k)}
\end{align}
where 
\begin{align*}
\boldsymbol{\mathit{V}}^{(k)} = & 
\begin{pmatrix}
\mathit{y}_{\mathit{p} + 1}^{(k)} - \mathit{x}_{\mathit{p} + 1} \\
\mathit{y}_{\mathit{p} + 2}^{(k)} - \mathit{x}_{\mathit{p} + 2} \\
\vdots \\
\mathit{y}_{n}^{(k)} - \mathit{x}_{n}
\end{pmatrix}
, \qquad
\phi^{(k)} = 
\begin{pmatrix}
\phi_1^{(k)} \\
\phi_2^{(k)} \\
\vdots \\
\phi_{\mathit{p}}^{(k)}
\end{pmatrix} 
, \\
\boldsymbol{\mathit{Z}}^{(k)} = & 
\begin{pmatrix}
\mathit{y}_{\mathit{p}}^{(k)} - \mathit{x}_{\mathit{p}}
 & \mathit{y}_{\mathit{p}-1}^{(k)} - \mathit{x}_{\mathit{p}-1}
 & \ldots 
 & \mathit{y}_1^{(k)} - \mathit{x}_1 \\
\mathit{y}_{\mathit{p}+1}^{(k)} - \mathit{x}_{\mathit{p}+1}
 & \mathit{y}_{\mathit{p}}^{(k)} - \mathit{x}_{\mathit{p}}
 & \ldots 
 & \mathit{y}_2^{(k)} - \mathit{x}_2 \\
\vdots & \vdots & \ddots & \vdots \\
\mathit{y}_{n-1}^{(k)} - \mathit{x}_{n-1} 
 & \mathit{y}_{n-2}^{(k)} - \mathit{x}_{n-2}
 & \ldots 
 & \mathit{y}_{n-\mathit{p}}^{(k)} - \mathit{x}_{n-\mathit{p}}
\end{pmatrix}. 
\end{align*}

\begin{example}[Parameter estimation on $\mathit{y}^{(0)}$, Example \ref{example:algorithm-overview} continued]\label{example:parameter-estimation}
Consider $\mathit{x} = \{6, 10, 9.6, 8.3, 7.7, 5.4, 5.6, 5.9, 6.3, 6.8, \allowbreak 7.5, 8.5\}$ and $\mathit{y}^{(0)} = \{6, 5.6, 5.4, 8.3, 7.7, 5.4, 5.6, 5.9, 6.3, 6.8, \allowbreak 7.5, 8.5\}$. 
Given order $\mathit{p}=1$, we have $\boldsymbol{\mathit{V}}^{(0)} = \{-4.4,-4.2,0,\allowbreak 0,0,\allowbreak 0,0,0,0,0,0\}'$ 
with 11 rows and 1 column,
and $\boldsymbol{\mathit{Z}}^{(0)} = \{0,-4.4, \allowbreak -4.2,0,0,0,0,0,0,0,0\}'$
with 11 rows and 1 column.
Referring to Equation \ref{equation:ols},
the parameter is estimated by
$$\phi_1^{(0)} = \frac{(-4.4)*(-4.2)}{(-4.4)^2+(-4.2)^2} = 0.5.$$
\end{example}

Owing to the iterative repairing,  online incremental parameter estimation is necessary, which is not studied in the previous studies (see our approach in Section \ref{sect:prune}).

\subsection{Candidate Generation}
\label{sect:candidate-generation}

The repair candidate generation step S2 (in Line \ref{algorithm:repair-candidate} in Algorithm \ref{algorithm:arx-iteration})
employs ARX($\mathit{p}$)  to infer the candidate repair 
$\hat{\mathit{y}}^{(k)} = \phi^{(k)}\cdot(\mathit{y}^{(k)} - \mathit{x}) + \mathit{x},$
referring to the estimated parameter $\phi^{(k)}$.  
More specifically,
for each point $t$, 
$\hat{\mathit{y}}_t^{(k)}$ is given~by  
\begin{align}
\label{equation:candidate}
\hat{\mathit{y}}_t^{(k)} = \sum_{i=1}^{p}\phi_i^{(k)}(\mathit{y}_{t-i}^{(k)}-\mathit{x}_{t-i}) + \mathit{x}_t
\end{align}
according to $\mathit{y}_{t-1}^{(k)}, \ldots, \mathit{y}_{t-p}^{(k)}$. 
We note that only candidates with $|\hat{\mathit{y}}_t^{(k)} - \mathit{y}_t^{(k)}| > \tau$ need to be considered 
referring to the convergence condition in Equation \ref{equation:convergence}.

\begin{example}[Repair candidate $\hat{\mathit{y}}^{(0)}$, Example \ref{example:parameter-estimation} continued]\label{example:repair-candidate}
Consider the parameter $\phi_1^{(0)} = 0.5$ estimated in Example \ref{example:parameter-estimation}. 
Let threshold $\tau = 0.1$. 
Referring to Equation \ref{equation:candidate}, 
we have 
$\hat{\mathit{y}}_{4}^{(0)} = 0.5 * (5.4 - 9.6) + 8.3 = 6.2$
with $|\hat{\mathit{y}}_{4}^{(0)} - \mathit{y}_{4}^{(0)}| = |6.2 - 8.3| = 2.1 > 0.1$, 
and 
$\hat{\mathit{y}}_{5}^{(0)} = 0.5 * (8.3 - 8.3) + 7.7 = 7.7$ 
with $|7.7 - 7.7| = 0 < 0.1$.
The repair candidates are $\hat{\mathit{y}}^{(0)}  = \{\textcolor{blue}{+}, \textcolor{blue}{+}, \textcolor{blue}{+}, 6.2, -, \textcolor{blue}{+}, -, -, -, -, -\textcolor{blue}{+}\}$ 
where 
`$\textcolor{blue}{+}$' corresponds to the labeled points and 
`$-$' denotes no candidates. 
That is, we need to consider only one candidate $\hat{\mathit{y}}_{4}^{(0)}$ for repairing.
\end{example}

\subsection{Repair Evaluation}
\label{sect:repair-evaluation}

The repair evaluation step S3  
(in Line \ref{algorithm:repair-evaluate} in Algorithm \ref{algorithm:arx-iteration})
selects one repair to accept, 
i.e., assigning $ \mathit{y}_t^{(k+1)} = \hat{\mathit{y}}_t^{(k)}$
the aforesaid generated repair candidate.
Following the minimum change principle in data repairing \cite{DBLP:conf/sigmod/BohannonFFR05}, 
the repair that minimally differs from its original input is preferred with higher confidence.
The repaired result in each iteration is: 
\begin{align}\label{equation:minimum-solution}
\mathit{y}_t^{(k+1)} = 
  \begin{cases}
\hat{\mathit{y}}_t^{(k)} & \quad \textrm{if } t = \argmin_{i} |\hat{\mathit{y}}_i^{(k)} - \mathit{x}_i| \\
\mathit{y}_t^{(k)} & \quad \text{otherwise }
\end{cases} .
\end{align}

Remarkably, only one data point with the minimum change (most confident) is repaired in each iteration, which is more efficient than the NP-hard problem of minimizing the overall changes w.r.t.\ integrity constraints~\cite{DBLP:conf/sigmod/BohannonFFR05}.%

\begin{example}[Minimum repair $\hat{\mathit{y}}_t^{(1)}$, Example \ref{example:repair-candidate} continued]
\label{example:repair-evaluation}
Since there is only one repair candidate obtained in Example \ref{example:repair-candidate}, 
i.e.,  $\hat{\mathit{y}}_{4}^{(0)} = 6.2$, it is the minimum repair (among all candidates).
The sequence after the first iteration becomes $\mathit{y}^{(1)} = \{6, 5.6, 5.4, {6.2}, 7.7, 5.4, 5.6, 5.9, 6.3, 6.8, 7.5, 8.5\}$.
\end{example}

Note that the minimum change principle \cite{DBLP:conf/sigmod/BohannonFFR05} in data repairing is based on the intuition that human or systems always \emph{try to} minimize their mistakes. 
However, it is not guaranteed that the minimum change repair always corresponds to the true value. 
Therefore, similar to other minimum change-based data repairing studies \cite{DBLP:conf/sigmod/BohannonFFR05,DBLP:conf/icde/ChuIP13}, the accuracy of the final results is unlikely to have theoretical guarantees, since there is no constraint on how far the errors may diverge from the truth.
For this reason, we can only evaluate the correctness of the proposed repair by comparing to the ground truth in experiments, similar to other data repairing studies \cite{DBLP:conf/sigmod/BohannonFFR05,DBLP:conf/icde/ChuIP13} as well. 
Nevertheless, 
we can show that the efficient pruning and incremental computation are safe (Propositions \ref{proposition:matrix-reduce} and \ref{the:incremental}), 
i.e., the accuracy of the final results with efficient computing is theoretically guaranteed to be the same as the results of original IMR without pruning and incremental computation.

%-------------------------------------------------------------------------
\section{Convergence Analysis}\label{sect:analysis}

In this section, we analyze the convergence of iterative repairing, i.e., 
$\lim_{k\rightarrow+\infty} \sum_{i=1}^n \left(
\mathit{y}_i^{(k+1)} - \mathit{y}_i^{(k)}
\right)
 = 0, $
which is essential to the termination of Algorithm \ref{algorithm:arx-iteration}. 
While the general convergence problem is still open, we study the convergence of the proposed method in certain special cases in this section for two reasons: 
(1) we illustrate that the ARX-based approach is indeed a special case of the proposed IMR with static parameter (Proposition \ref{proposition:fixed-phi-one-round}) in Section \ref{sect:parameter-static}; 
(2) we identify another special case with convergence guarantee in Section \ref{sect:parameter-converged}, which enables online repairing over streaming data without iteration 
(also see Sections \ref{sect:special-one-label} and \ref{sect:online-exp} for more details and experiments).

\subsection{Static Parameter}
\label{sect:parameter-static}

We study this special case in order to illustrate the relationship between our proposed IMR and the existing ARX.
Let us first analyze the convergence of IMR (Proposition \ref{proposition:mean-converge-fixed}) 
and then illustrate their equivalence in certain case (Proposition \ref{proposition:fixed-phi-one-round}). 

Rather than dynamically updating the parameter $\phi^{(k)}$ in each iteration, 
in Line \ref{algorithm:repair-parameter} in Algorithm \ref{algorithm:arx-iteration}, 
a special case is to specify a static parameter, $\phi^{(k)}=\phi^{(0)}$, for all the iterations.

\begin{proposition}
\label{proposition:mean-converge-fixed}
With a static parameter $\phi^{(k)} = \phi, \forall k$, 
the repair result converges, 
i.e., 
$$\lim_{k\rightarrow+\infty} \sum_{i=1}^n \left(
\mathit{y}_i^{(k+1)} - \mathit{y}_i^{(k)}
\right)
 = 0.$$
\end{proposition}

\begin{proof}
Let $t$ be the point repaired in iteration $k$. 
We show that $\mathit{y}_t^{(k+i)}=\mathit{y}_t^{(k+1)}$ converges (will not change in the following iterations $k+i$), 
if the $\mathit{p}$ preceding points of $t$ are converged (no longer change). 
The reason is that $\hat{\mathit{y}_t}^{(k+i)}=\mathit{y}_t^{(k+1)}$,
referring to Equation \ref{equation:candidate} of candidate generation 
with static parameter $\phi$.
That is, no repair candidate will be generated.

If one of the $\mathit{p}$ preceding points of $t$ changes (not converged), say $t'$, 
the aforesaid derivation applies similarly to $t'$. 
By recursively applying the derivations, we can finally find a point $t^*$ whose $\mathit{p}$ preceding points are either labeled or converged, i.e., a new point $t^*$ converges. 
Such $t^*$ always exists, i.e., $\mathit{p}+1$, given that the first $\mathit{p}$ points in the sequence will not be modified by the repairing 
(referring to the repair candidate generation in Equation \ref{equation:candidate}) 
and thus already converged under the static parameter $\phi$.

Finally, all the data points gradually converge. 
\end{proof}

\subsubsection*{Special Case of IMR(1)}

We show in the following that
for $\mathit{p}=1$,  
the ARX($\mathit{p}$)-based repairing in Section \ref{sect:ARX} 
is a special case of our proposed IMR($\mathit{p}$) with static parameter $\phi^{(0)}$.
This equivalence demonstrates the rationale of our proposal.

\begin{proposition}
\label{proposition:fixed-phi-one-round}
For IMR(1) with static parameter $\phi^{(k)} = \phi, \forall k$, 
Algorithm \ref{algorithm:arx-iteration} 
is equivalent to 
ARX(1)-based repairing.
\end{proposition}

\begin{proof}
Suppose that the first point is labeled. 
Consider data point 2. 
Referring to Proposition \ref{proposition:mean-converge-fixed}, 
$\mathit{y}_2^{(k_2)}$ becomes converged in some iteration $k_2$, 
denoted by $\mathit{y}_2=\phi_1(\mathit{y}_1-\mathit{x}_1)+\mathit{x}_2$, 
where $\mathit{y}_1=\mathit{y}_1^{(0)}$ is labeled. 

Similarly, given the converged $\mathit{y}_2$, 
$\mathit{y}_3^{(k_3)}$ will converge in some iteration $k_3$, 
denoted by $\mathit{y}_3=\phi_1(\mathit{y}_2-\mathit{x}_2)+\mathit{x}_3$. 

By recursively obtaining the converged  $\mathit{y}_t$ which is not labeled, 
it is exactly the procedure of ARX(1)-based repairing. 
\end{proof}

\subsection{Converged Parameter}
\label{sect:parameter-converged}

We now consider the dynamically updated parameter $\phi^{(k)}$ in each iteration, 
in Line \ref{algorithm:repair-parameter} in Algorithm \ref{algorithm:arx-iteration}. 
As shown in the following Proposition \ref{proposition:phi-y-converge}, if the dynamic parameter converges, 
the repair converges as well. 
This converged parameter case is interesting, since 
the corresponding converged repair result could be directly calculated without iterative computing in certain cases
as illustrated below. 

\begin{proposition}
\label{proposition:phi-y-converge}
If the parameter converges, 
$\lim_{k\rightarrow+\infty} \phi^{(k)} = \phi,$ 
then the repair also converges
$$\lim_{k\rightarrow+\infty} \sum_{i=1}^n \left(
\mathit{y}_i^{(k+1)} - \mathit{y}_i^{(k)}
\right)
 = 0.$$
\end{proposition}

\begin{proof}
When the parameter $\phi$ converges, it becomes static. 
Similar proof for Proposition \ref{proposition:mean-converge-fixed} applies.
\end{proof}

\subsubsection{Special Case of IMR(1)}
\label{sect:special-converge-1}

Again, we consider the special case of IMR(1) with order $\mathit{p}=1$. 
To show how the repair results could be directly computed without iterations, we first illustrate that any $\mathit{y}_t^{(k)}$ generated during Algorithm \ref{algorithm:arx-iteration} can be represented as follows, 
a.k.a. provenance of $\mathit{y}_t^{(k)}$. 

\begin{lemma}
\label{proposition:provenance-ar1}
For IMR(1), we can represent each $\mathit{y}_t^{(k)}$ by 
$$\mathit{y}_t^{(k)} = \phi_1^{(k_s)}\phi_1^{(k_{s-1})}\dots\phi_1^{(k_1)}(\mathit{y}_{t-s}^{(0)}-\mathit{x}_{t-s}) + \mathit{x}_t,$$
where $0<k_1<\dots<k_{s-1}<k_s<k$ denote iteration numbers, 
$\mathit{y}_{t-s}^{(0)}$ is labeled truth, 
and 
time points $t-s+1, \dots, t$ are not labeled. 
\end{lemma}

\begin{proof}
Referring to Equation \ref{equation:minimum-solution}, $\mathit{y}_t^{(k)}$ in $k$-th iteration is either unchanged (equal to $\mathit{y}_t^{(k-1)}$)
or determined by $\hat{\mathit{y}}_t^{(k-1)}=\phi_1^{(k-1)}(\mathit{y}_{t-1}^{(k-1)}-\mathit{x}_{t-1}) + \mathit{x}_t$ referring to Equation \ref{equation:candidate} with $\mathit{p}=1$.

Suppose that $\mathit{y}_t^{(k)}$ is unchanged since iteration $k_s$.
We have 
\begin{align*}
\mathit{y}_t^{(k)} = \mathit{y}_t^{(k-1)} = \ldots = \mathit{y}_t^{(k_s+1)} 
= \phi_1^{(k_s)}(\mathit{y}_{t-1}^{(k_s)}-\mathit{x}_{t-1}) + \mathit{x}_t 
.
\end{align*}
Similar provenance applies to $\mathit{y}_{t-1}^{(k_s)}, \dots, \mathit{y}_{t-1}^{(k_1)}$, 
\begin{align*}
\mathit{y}_{t-1}^{(k_s)} 
= & \dots = \mathit{y}_{t-1}^{(k_{s-1}+1)} 
= \phi_1^{(k_{s-1})}(\mathit{y}_{t-2}^{(k_{s-1})}-\mathit{x}_{t-2}) + \mathit{x}_{t-1} 
, \\
& \dots ,\\
\mathit{y}_{t-s+1}^{(k_2)} 
= & \dots = \mathit{y}_{t-s+1}^{(k_1+1)} 
= \phi_1^{(k_1)}(\mathit{y}_{t-s}^{(k_1)}-\mathit{x}_{t-s}) + \mathit{x}_{t-s+1} 
, \\
\mathit{y}_{t-s}^{(k_1)} 
= & \dots = \mathit{y}_{t-s}^{(0)}
.
\end{align*}
Combining these $s+1$ derivations, the conclusion is proved. 
\end{proof}

We denote the labeled points in $\mathit{y}^{(0)}$ by multiple (say $m$) segments. 
Let $s(j)$ and $e(j)$ denote the start and end point of the $j$-th labeled segment, $j=1,\dots,m$.
For instance, there are 3 segments of labeled data points in Figure \ref{fig:example-dirty}, 
having $
s(1)=1, e(1)=2, 
s(2)=4, e(2)=5, 
s(3)=9, e(3)=9.
$

Let $\mathit{z}_i = \mathit{y}_i^{(0)} - \mathit{x}_i$ for all labeled points $i$.
(We set $\mathit{z}_{e(0)}=0$ over undefined segment 0.)

\begin{proposition}
\label{proposition:dynamic-ar1-general}
For IMR(1), if the parameter converges, having 
$\lim_{k\rightarrow+\infty} \phi_1^{(k)} = \phi_1,$
then the converged repair result can be directly given by
$$\lim_{k\rightarrow+\infty}\mathit{y}_i^{(k)} = \mathit{y}_i, $$ 
where
\begin{align}
\label{equation:y-converge-mulitple}
\mathit{y}_i =  
\begin{cases}
\mathit{y}_i^{(0)}
&\mathrm{if~} i \in [s(j),e(j)] \\
\phi_1^{i-e(j)}(\mathit{y}_{e(j)}^{(0)} - \mathit{x}_{e(j)}) + \mathit{x}_i
&\mathrm{if~} i \in (e(j), s(j+1))
\end{cases} 
\end{align}
and the converged parameter $\phi_1$ is a solution to 
\begin{align}
\label{equation:parameter-converge-mulitple}
\phi_1 = \frac{\sum\limits_{j=1}^{m}\left(
\phi_1^{s(j)-1-e(j-1)} \mathit{z}_{e(j-1)} \mathit{z}_{s(j)}
+ 
\sum\limits_{i=s(j)}^{e(j)-1}\mathit{z}_{i}\mathit{z}_{i+1} 
\right)
}{
\sum\limits_{j=1}^{m}\left(
(\phi_1^{s(j)-1-e(j-1)} \mathit{z}_{e(j-1)})^2
+ 
\sum\limits_{i=s(j)}^{e(j)-1}(\mathit{z}_{i})^2 
\right)
}.
\end{align}
\end{proposition}
\begin{proof} 
Given $\lim\limits_{k\rightarrow+\infty} \phi_1^{(k)} = \phi_1$,
according to Proposition \ref{proposition:phi-y-converge}, the result $\mathit{y}^{(k)}$ will converge, having $\lim\limits_{k\rightarrow+\infty}\mathit{y}_i^{(k)} = \mathit{y}_i$.

Referring to Lemma \ref{proposition:provenance-ar1},
the converged value has its provenance from the nearest preceding labeled point, i.e., $\mathit{y}_i = \phi_1^{i-e(j)}(\mathit{y}_{e(j)}^{(0)} - \mathit{x}_{e(j)}) + \mathit{x}_i$, for $e(j) < i < s(j+1)$.

Let $\mathit{z}_t = \mathit{y}_t - \mathit{x}_t$ for all unlabeled points $t$.
Referring to Equation \ref{equation:ols} of parameter estimation, we have
\begin{align*}
\phi_1 &= \frac{
\mathit{z}_1\mathit{z}_2
+\ldots+
\mathit{z}_{n-1}\mathit{z}_{n}
}{
(\mathit{z}_1)^2
+\ldots+
(\mathit{z}_{n-1})^2
}. 
\end{align*}

Similar to the aforesaid $\mathit{y}_i$, we have  
$\mathit{y}_{i+1}= \phi_1^{i-e(j)+1}(\mathit{y}_{e(j)}^{(0)} - \mathit{x}_{e(j)}) + \mathit{x}_{i+1}$, for $e(j) < i+1 < s(j+1)$.
It follows 
$\mathit{z}_{i+1}=\phi_1\mathit{z}_i.$

By applying $\mathit{z}_{i+1}=\phi_1\mathit{z}_i$ for all $e(j) \leq i < s(j+1)-1, j \in [1, m]$, we have 
\begin{align*}
\phi_1 &= \frac{
\mathit{z}_1\mathit{z}_2
+\ldots+
\mathit{z}_{n-1}\mathit{z}_{n}
}{
(\mathit{z}_1)^2
+\ldots+
(\mathit{z}_{n-1})^2
} \\
 &= \frac{
\sum\limits_{j=1}^{m}\left(
\sum\limits_{i=e(j-1)}^{s(j)-1}\mathit{z}_{i}\mathit{z}_{i+1} 
+ 
\sum\limits_{i=s(j)}^{e(j)-1}\mathit{z}_{i}\mathit{z}_{i+1} 
\right)
}{
\sum\limits_{j=1}^{m}
\left(
\sum\limits_{i=e(j-1)}^{s(j)-1}(\mathit{z}_{i})^2
+ 
\sum\limits_{i=s(j)}^{e(j)-1}(\mathit{z}_{i})^2
\right)
} \\
 &= \frac{
\sum\limits_{j=1}^{m}\left(
\sum\limits_{i=e(j-1)}^{s(j)-2}\phi_1(\mathit{z}_{i})^2 
+ 
\mathit{z}_{s(j)-1} \mathit{z}_{s(j)}
+
\sum\limits_{i=s(j)}^{e(j)-1}\mathit{z}_{i}\mathit{z}_{i+1} 
\right)
}{
\sum\limits_{j=1}^{m}
\left(
\sum\limits_{i=e(j-1)}^{s(j)-1}(\mathit{z}_{i})^2
+ 
\sum\limits_{i=s(j)}^{e(j)-1}(\mathit{z}_{i})^2
\right)
} \\
 &= \frac{
\sum\limits_{j=1}^{m}\left(
\mathit{z}_{s(j)-1} \mathit{z}_{s(j)}
+ 
\sum\limits_{i=s(j)}^{e(j)-1}\mathit{z}_{i}\mathit{z}_{i+1} 
\right)
}{
\sum\limits_{j=1}^{m}
\left(
(\mathit{z}_{s(j)-1})^2 
+ \sum\limits_{i=s(j)}^{e(j)-1}(\mathit{z}_{i})^2
\right)
} 
\end{align*}

By applying Equation \ref{equation:y-converge-mulitple} again, i.e., 
$$\mathit{z}_{s(j)-1}=\phi_1^{s(j)-1-e(j-1)} \mathit{z}_{e(j-1)},$$ 
the conclusion in Equation \ref{equation:parameter-converge-mulitple} is proved.
\end{proof}

\subsubsection{IMR(1) with One Labeled Segment}
\label{sect:special-one-label}

We consider the case that only one segment with length $\ell$ is labeled at the beginning of $\mathit{y}^{(0)}$, 
i.e., $\mathit{y}^{(0)}_1, \mathit{y}^{(0)}_2,  \dots,  \mathit{y}^{(0)}_\ell$ are labeled. 
In this special case, the converged parameter and repair result can be directly calculated  without iterating, 
and most importantly 
it enables efficient online computation, 
by interpreting all the historical data as one segment labeled (see Section \ref{sect:online-exp} for details and evaluation). 
Remarkably, no threshold needs to be set in this case.

The idea is: 
(1) We first show in Lemma \ref{lemma:bound-condition} that under certain inputs, the estimated parameter in each iteration is indeed bounded;
(2) Proposition \ref{proposition:mean-converge-dynamic} then illustrates that the bounded parameter leads to converged parameter;
(3) Finally, analogous to Proposition \ref{proposition:dynamic-ar1-general}, given the converged parameter, 
we directly calculate the converged repair without iterative computing in Proposition \ref{proposition:dynamic-ar1-one}.

\begin{lemma}
\label{lemma:bound-condition}
For IMR(1) with first $\ell$ data points labeled in $\mathit{y}^{(0)}$.
If the input satisfies
$\left|\sum\limits_{t=1}^{\ell-1}\mathit{z}_{t}^{(0)}\mathit{z}_{t+1}^{(0)}\right|
< \sum\limits_{t=1}^{\ell-1}\mathit{z}_t^{(0)}\mathit{z}_t^{(0)}$, i.e., 
$$\left|\sum\limits_{t=1}^{\ell-1}(\mathit{y}_{t}^{(0)}-\mathit{x}_{t})(\mathit{y}_{t+1}^{(0)}-\mathit{x}_{t+1})\right|
< \sum\limits_{t=1}^{\ell-1}(\mathit{y}_t^{(0)}-\mathit{x}_t)^2,$$
then we have $|\phi_1^{(k)}| < 1$ in the iterations $k, 0 \leq k \leq n-\ell$.
\end{lemma}
\begin{proof}

According to Equation \ref{equation:ols}, 
we have 
$$\phi_1^{(0)} = \frac{
\sum\limits_{t=1}^{\ell-1}\mathit{z}_{t}^{(0)}\mathit{z}_{t+1}^{(0)}
}{
(\mathit{z}_{\ell}^{(0)})^2+\sum\limits_{t=1}^{\ell-1}(\mathit{z}_{t}^{(0)})^2
}.$$ 
Referring to the given condition, it follows
$$
|\phi_1^{(0)}| \leq 
\frac{
|\sum\limits_{t=1}^{\ell-1}\mathit{z}_{t}^{(0)}\mathit{z}_{t+1}^{(0)}|
}{
\sum\limits_{t=1}^{\ell-1}(\mathit{z}_{t}^{(0)})^2
}
<1.$$

We prove $|\phi_1^{(k)}| < 1, 1 \leq k \leq n-\ell$ in three cases.

\paragraph*{Case 1: $\phi_1^{(0)} > 0$}

We will prove $0<\phi_1^{(k)}< 1, 1 \leq k \leq n-\ell$ by induction. 

\noindent\textbf{Basis:}
For $k = 1$, 
referring to the minimal change principle and the candidate generation in Section \ref{sect:candidate-generation}, 
we have $\argmin\limits_{i \in [1,n]}|\hat{\mathit{z}}_i^{(0)}|=\ell+1$ and  
$\mathit{z}_{\ell+1}^{(1)} = \hat{\mathit{z}}_{\ell+1}^{(0)} = \mathit{z}_{\ell}^{(0)}\phi_1^{(0)}$. Since only one point will be changed in each iteration, we have
$\mathit{z}_i^{(1)} = \mathit{z}_i^{(0)}, i \neq \ell+1$.
Referring to Equation \ref{equation:ols}, 
it follows
\begin{align*}
\phi_1^{(1)} &= \frac{
\mathit{z}_{\ell+1}^{(1)}\mathit{z}_{\ell}^{(1)} + \sum\limits_{t=1}^{\ell-1}\mathit{z}_{t}^{(0)}\mathit{z}_{t+1}^{(0)}
}{
(\mathit{z}_{\ell+1}^{(1)})^2 + \sum\limits_{t=1}^{\ell}(\mathit{z}_{t}^{(0)})^2
}
\end{align*}

The following derivation shows $\phi_1^{(1)} > 0$.
\begin{align*}
\phi_1^{(1)}-\phi_1^{(0)} &= 
\frac{
\mathit{z}_{\ell+1}^{(1)}\mathit{z}_{\ell}^{(1)} + \sum\limits_{t=1}^{\ell-1}\mathit{z}_{t}^{(0)}\mathit{z}_{t+1}^{(0)}
}{
(\mathit{z}_{\ell+1}^{(1)})^2 + \sum\limits_{t=1}^{\ell}(\mathit{z}_{t}^{(0)})^2
} - \frac{
\sum\limits_{t=1}^{\ell-1}\mathit{z}_{t}^{(0)}\mathit{z}_{t+1}^{(0)}
}{
\sum\limits_{t=1}^{\ell}(\mathit{z}_{t}^{(0)})^2
} \\
&= \frac{
\mathit{z}_{\ell+1}^{(1)}(\mathit{z}_{\ell}^{(1)}(\sum\limits_{t=1}^{\ell}(\mathit{z}_{t}^{(0)})^2)
-\mathit{z}_{\ell+1}^{(1)}(\sum\limits_{t=1}^{\ell-1}\mathit{z}_{t}^{(0)}\mathit{z}_{t+1}^{(0)}))
}{
(\sum\limits_{t=1}^{\ell}(\mathit{z}_{t}^{(0)})^2)((\mathit{z}_{\ell+1}^{(1)})^2 + \sum\limits_{t=1}^{\ell}(\mathit{z}_{t}^{(0)})^2)
} \\
&= \frac{
\phi_1^{(0)}(1-(\phi_1^{(0)})^2)(\mathit{z}_{\ell}^{(0)})^2
}{
(\mathit{z}_{\ell+1}^{(1)})^2 + \sum\limits_{t=1}^{\ell}(\mathit{z}_{t}^{(0)})^2
} > 0
\end{align*}

And it follows $\phi_1^{(1)} < 1$ below.
\begin{align*}
1-\phi_1^{(1)} &= \frac{
(\sum\limits_{t=1}^{\ell}(\mathit{z}_{t}^{(0)})^2-\sum\limits_{t=1}^{\ell-1}\mathit{z}_{t}^{(0)}\mathit{z}_{t+1}^{(0)}) 
+\mathit{z}_{\ell+1}^{(1)}(\mathit{z}_{\ell+1}^{(1)}-\mathit{z}_{\ell}^{(1)})
}{
(\mathit{z}_{\ell+1}^{(1)})^2 + \sum\limits_{t=1}^{\ell}(\mathit{z}_{t}^{(0)})^2
} \\
&= \frac{
(\sum\limits_{t=1}^{\ell}(\mathit{z}_{t}^{(0)})^2-\sum\limits_{t=1}^{\ell-1}\mathit{z}_{t}^{(0)}\mathit{z}_{t+1}^{(0)}) 
+(\mathit{z}_{\ell}^{(0)})^2(\phi_1^{(0)}(\phi_1^{(0)}-1))
}{
(\mathit{z}_{\ell+1}^{(1)})^2 + \sum\limits_{t=1}^{\ell}(\mathit{z}_{t}^{(0)})^2
} \\
&> \frac{
(\mathit{z}_{\ell}^{(0)})^2(1+(\phi_1^{(0)})^2 - \phi_1^{(0)})
}{
(\mathit{z}_{\ell+1}^{(1)})^2 + \sum\limits_{t=1}^{\ell}(\mathit{z}_{t}^{(0)})^2
} > 0
\end{align*}

Combining the aforesaid two derivations, we have $0 < \phi_1^{(0)} < \phi_1^{(1)} < 1$.

\noindent\textbf{Inductive step:}
We will show that if the conclusion holds $0 < \phi_1^{(k)} < 1$, 
for $k = i-1, 1 \leq i \leq n-\ell$, 
then also $0 < \phi_1^{(i)} < 1$ holds for $k=i$.

Following the IMR, we have
$\argmin\limits_{j \in [1,n]}|\hat{\mathit{z}}_j^{(i-1)}|
=\ell+i$ and  
$\mathit{z}_{\ell+i}^{(i)} = \hat{\mathit{z}}_{\ell+i}^{(i-1)} = \mathit{z}_{\ell+i-1}^{(i-1)}\phi_1^{(i-1)}$.

The following shows $\phi_1^{(i)} > 0$.
\begin{align*}
&\phi_1^{(i)}-\phi_1^{(i-1)} \allowdisplaybreaks \\ 
&= \frac{
\mathit{z}_{\ell+i}^{(i)}\mathit{z}_{\ell+i-1}^{(i)} + \sum\limits_{t=1}^{\ell+i-2}\mathit{z}_{t}^{(i-1)}\mathit{z}_{t+1}^{(i-1)}
}{
(\mathit{z}_{\ell+i}^{(i)})^2 + \sum\limits_{t=1}^{\ell+i-1}(\mathit{z}_{t}^{(i-1)})^2
} - \frac{
\sum\limits_{t=1}^{\ell+i-2}\mathit{z}_{t}^{(i-1)}\mathit{z}_{t+1}^{(i-1)}
}{
\sum\limits_{t=1}^{\ell+i-1}(\mathit{z}_{t}^{(i-1)})^2
} \allowdisplaybreaks \\ 
&= \frac{
\mathit{z}_{\ell+i}^{(i)}
(\mathit{z}_{\ell+i-1}^{(i)}(\sum\limits_{t=1}^{\ell+i-1}(\mathit{z}_{t}^{(i-1)})^2)
-\mathit{z}_{\ell+i}^{(i)}(\sum\limits_{t=1}^{\ell+i-2}\mathit{z}_{t}^{(i-1)}\mathit{z}_{t+1}^{(i-1)}))
}{
(\sum\limits_{t=1}^{\ell+i-1}(\mathit{z}_{t}^{(i-1)})^2)((\mathit{z}_{\ell+i}^{(i)})^2 + \sum\limits_{t=1}^{\ell+i-1}(\mathit{z}_{t}^{(i-1)})^2)
} \\
&= \frac{
\phi_1^{(i-1)}(1-(\phi_1^{(i-1)})^2)(\mathit{z}_{\ell+i-1}^{(i-1)})^2
}{
(\mathit{z}_{\ell+i}^{(i)})^2 + \sum\limits_{t=1}^{\ell+i-1}(\mathit{z}_{t}^{(i-1)})^2
} > 0
\end{align*}

It follows $\phi_1^{(i)} < 1$ as illustrated below.
\begin{align*}
& 
\left(1-\phi_1^{(i)}\right)\left( 
(\mathit{z}_{\ell+i}^{(i)})^2 + \sum\limits_{t=1}^{\ell+i-1}(\mathit{z}_{t}^{(i-1)})^2
\right) \\
&= (\sum\limits_{t=1}^{\ell}(\mathit{z}_{t}^{(0)})^2-\sum\limits_{t=1}^{\ell-1}\mathit{z}_{t}^{(0)}\mathit{z}_{t+1}^{(0)}) \\  
&+(\mathit{z}_{\ell}^{(0)})^2 \times 
\sum\limits_{s=1}^{i}\prod\limits_{j=0}^{s-2}(\phi_1^{(j)})^2\phi_1^{(s-1)}(\phi_1^{(s-1)}-1) \\
&> (\mathit{z}_{\ell}^{(0)})^2 \times
(1 + (\phi_1^{(i-1)}-1)(\sum\limits_{s=1}^{i}(\phi_1^{(i-1)})^{2s-1})) \\
&= (\mathit{z}_{\ell}^{(0)})^2 \times
(1+ (\phi_1^{(i-1)}-1)\frac{
\phi_1^{(i-1)}(1-(\phi_1^{(i-1)})^{2i})
}{
1-(\phi_1^{(i-1)})^2
}) \\
&= (\mathit{z}_{\ell}^{(0)})^2 \times
(1 + (\phi_1^{(i-1)}-1)\frac{
\phi_1^{(i-1)}-(\phi_1^{(i-1)})^{2i+1}
}{
(1-\phi_1^{(i-1)})(1+\phi_1^{(i-1)})
}) \\
&= (\mathit{z}_{\ell}^{(0)})^2 \times
\frac{
1+(\phi_1^{(i-1)})^{2i+1}
}{
1+\phi_1^{(i-1)}
} > 0
\end{align*}

Thereby $0 < \phi_1^{(k)} < 1$ holds for $k=i$.

\paragraph*{Case 2: $\phi_1^{(0)} < 0$}
Similar proof steps apply.

\paragraph*{Case 3: $\phi_1^{(0)} = 0$}
We can show that $\phi_1^{(k)} = 0, 1\leq k \leq n-\ell$.
\end{proof}

The following conclusion illustrates that with a bounded parameter in each iteration,
the parameter converges. 

\begin{proposition}
\label{proposition:mean-converge-dynamic}
For IMR(1) with first $\ell$ data points labeled in $\mathit{y}^{(0)}$, if $|\phi_1^{(k)}| < 1$ in the iterations $k, 0 \leq k \leq n-\ell$,
then the parameter converges, 
i.e., 
$$\lim_{k\rightarrow+\infty} \phi_1^{(k)} = \phi_1,$$
\end{proposition}

\begin{proof}
Let $\mathit{z}_t^{(k)} = \mathit{y}_t^{(k)} - \mathit{x}_t$ and 
$\hat{\mathit{z}}_t^{(k)} = \hat{\mathit{y}}_t^{(k)} - \mathit{x}_t$.
Suppose that point $t$ is selected to repair in the $k$-th iteration, 
having 
$$\mathit{z}_t^{(k+1)} =\hat{\mathit{z}}_t^{(k)}= \phi_1^{(k)}\mathit{z}_{t-1}^{(k)}$$
where $t  = \argmin\limits_{i \in [1,n]}|\hat{\mathit{z}}_i^{(k)}|$ 
referring to the minimum change.
Since one and only one point will be changed in each iteration, we have
$\mathit{z}_i^{(k+1)} = \mathit{z}_i^{(k)}, i \neq t$.

For the $(k + 1)$-th iteration, any point $i$ has
$$\hat{\mathit{z}}_i^{(k+1)}= \phi_1^{(k+1)}\mathit{z}_{i-1}^{(k+1)}.$$
It follows 
\begin{align*}
|\hat{\mathit{z}}_{t+1}^{(k+1)}| 
&= |\phi_1^{(k+1)}\mathit{z}_{t}^{(k+1)}| 
= |\phi_1^{(k+1)}|
|\phi_1^{(k)}|
|\mathit{z}_{t-1}^{(k)}| \\
&< |\phi_1^{(k+1)}|
|\mathit{z}_{t-1}^{(k)}|
= |\hat{\mathit{z}}_{t}^{(k+1)}|  \qquad (\text{since }|\phi_1^{(k)}| < 1) \\
&< |\phi_1^{(k+1)}|
|\mathit{z}_{i}^{(k+1)}|
= |\hat{\mathit{z}}_{i+1}^{(k+1)}|
\end{align*}
where $i \neq t$. 
That is, $(t+1)  = \argmin\limits_{i \in [1,n]}|\hat{\mathit{z}}_i^{(k+1)}|$,
will be the minimum repair in the $(k+1)$-th iteration.

Similar repairing applies in $m = n - \ell$ iterations,
having $\mathit{z}_n^{(m)} = \phi_1^{(1)}\ldots\phi_1^{(m-1)}\mathit{z}_{\ell}^{(0)}$, 
$\mathit{z}_{n-1}^{(m)} = \phi_1^{(1)}\ldots\phi_1^{(m-2)}\mathit{z}_{\ell}^{(0)}$, 
according to Lemma \ref{proposition:provenance-ar1}.
For the estimation of $\hat{\mathit{z}}_i^{(m)}$, 
we have $n = \argmin\limits_{i \in [1,n]}|\hat{\mathit{z}}_i^{(m)}|$ 
since $|\mathit{z}_{n-1}^{(m)}| = \min\limits_{i \in [1,n-1]}|\mathit{z}_{i}^{(m)}|$.
That is, point $n$ will always be  repaired in the following iterations.

According to Equation \ref{equation:ols}, we have
\begin{align*}
|\phi_1^{(m+1)} - \phi_1^{(m)}|
&= |\frac{
\mathit{z}_1^{(m+1)}\mathit{z}_2^{(m+1)}
+\ldots+
\mathit{z}_{n-1}^{(m+1)}\mathit{z}_{n}^{(m+1)}
}{
(\mathit{z}_1^{(m+1)})^2
+\ldots+
(\mathit{z}_{n-1}^{(m+1)})^2
} \\ 
& \qquad - \frac{
\mathit{z}_1^{(m)}\mathit{z}_2^{(m)}
+\ldots+
\mathit{z}_{n-1}^{(m)}\mathit{z}_{n}^{(m)}
}{
(\mathit{z}_1^{(m)})^2
+\ldots+
(\mathit{z}_{n-1}^{(m)})^2
}| \\
& = |\frac{
\mathit{z}_{n-1}^{(m)}
(\mathit{z}_{n}^{(m+1)}-\mathit{z}_{n}^{(m)})
}{
(\mathit{z}_1^{(m)})^2
+\ldots+
(\mathit{z}_{n-1}^{(m)})^2
}| \\
& = |\frac{
(\mathit{z}_{n-1}^{(m)})^2
}{
(\mathit{z}_1^{(m)})^2
+\ldots+
(\mathit{z}_{n-1}^{(m)})^2
}|\cdot
|\phi_1^{(m)}-\phi_1^{(m-1)}| \\
& < |\phi_1^{(m)}-\phi_1^{(m-1)}|
\end{align*}
where $\mathit{z}_{i}^{(m)} = \mathit{z}_{i}^{(m+1)}, i < n$. 
It is proved that the parameter converges.
\end{proof}

It is worth noting that the condition $|\phi_1^{(k)}| < 1$ in Proposition \ref{proposition:mean-converge-dynamic} for the parameter to converge 
could be commonly observed in real data. 
First, referring to \cite{brockwell2006introduction}, most time series in practice are stationary, 
which is guaranteed to have $|\phi_1^{(k)}| < 1$ for $\mathit{p}=1$. 
Moreover, for non-stationary cases, a typical processing way is to transform it to stationary via differencing \cite{brockwell2006introduction}.

\begin{proposition}
\label{proposition:dynamic-ar1-one}
For IMR(1) with first $\ell$ data points labeled in $\mathit{y}^{(0)}$, 
if the parameter converges, having 
$\lim_{k\rightarrow+\infty} \phi_1^{(k)} = \phi_1,$
then the converged repair result is
\begin{align}\label{equation:repair-converge-single}
\lim_{k\rightarrow+\infty}\mathit{y}_t^{(k)} = \mathit{y}_t = 
\begin{cases}
\mathit{y}_i^{(0)} & \quad \mathrm{if~} i \in [1,\ell] \\
\phi_1^{i-\ell}(\mathit{y}_{\ell}^{(0)} - \mathit{x}_{\ell}) + \mathit{x}_i & \quad \mathrm{if~} i > \ell
\end{cases} 
\end{align}
where the converge parameter can be directly calculated by 
\begin{align}
\label{equation:parameter-converge-single}
\phi_1 = \frac{
(\mathit{y}_1^{(0)}-\mathit{x}_1)(\mathit{y}_2^{(0)}-\mathit{x}_2)
+\dots+
(\mathit{y}_{\ell-1}^{(0)}-\mathit{x}_{\ell-1})(\mathit{y}_\ell^{(0)}-\mathit{x}_\ell)
}{
(\mathit{y}_1^{(0)}-\mathit{x}_1)^2
+\dots+
(\mathit{y}_{\ell-1}^{(0)}-\mathit{x}_{\ell-1})^2
}.
\end{align}
\end{proposition}

\begin{proof}
It is a special case of Proposition \ref{proposition:dynamic-ar1-general} with $m = 1$.
Similar to the proof of Proposition \ref{proposition:dynamic-ar1-general}, 
we have
\begin{align*}
\phi_1 &= \frac{
\mathit{z}_1\mathit{z}_2
+\ldots+
\mathit{z}_{n-1}\mathit{z}_{n}
}{
(\mathit{z}_1)^2
+\ldots+
(\mathit{z}_{n-1})^2
} \\
&= \frac{
\mathit{z}_1\mathit{z}_2
+\ldots+
\mathit{z}_{\ell-1}\mathit{z}_{\ell}
+ \phi_1(\mathit{z}_{\ell})^2
+ \ldots + \phi_1^{2n-2\ell-1}(\mathit{z}_{\ell})^2
}{
(\mathit{z}_1)^2
+\ldots+
(\mathit{z}_{\ell})^2
+ \phi_1^2(\mathit{z}_{\ell})^2
+\ldots+
\phi_1^{2n-2\ell-2}(\mathit{z}_{\ell})^2
} \\
&= \frac{
\mathit{z}_1\mathit{z}_2
+\ldots+
\mathit{z}_{\ell-1}\mathit{z}_{\ell}
}{
(\mathit{z}_1)^2
+\ldots+
(\mathit{z}_{\ell-1})^2
} .
\end{align*}
\end{proof}

It is not surprising that the converged parameter $\phi_1$ 
in Equation \ref{equation:parameter-converge-single}
in Proposition \ref{proposition:dynamic-ar1-one}
is exactly the solution of 
Equation \ref{equation:parameter-converge-mulitple}  
in Proposition \ref{proposition:dynamic-ar1-general}
with $m=1$
for the special case of one labeled segment.  

%-------------------------------------------------------------------------
\section{Efficient Parameter Estimation}\label{sect:prune}

Among the three major steps in Algorithm \ref{algorithm:arx-iteration}, 
while the repair candidate generation and evaluation (in Sections \ref{sect:candidate-generation} and \ref{sect:repair-evaluation}) are inevitable for the minimum repair,
we show in this section that the costly parameter $\phi^{(k)}$ estimation (in Section \ref{sect:parameter-estimation}) is optimizable 
in our iterative repairing scenario. 
First, we identify that the matrices $\boldsymbol{\mathit{Z}}^{(k)},\boldsymbol{\mathit{V}}^{(k)}$ for parameter estimation could be pruned by simply removing rows with value 0, in Section \ref{sect:matrix-pruning}. 
Moreover, incremental computation could be designed such that the complexity of parameter estimation is reduced from $O(n)$ to $O(1)$, in Section \ref{sect:incremental-computation}.

\subsection{Matrix Pruning}
\label{sect:matrix-pruning}

\paragraph*{Intuition}
Recall that when estimating the parameter $\phi^{(k)}$ in Equation \ref{equation:ols}, 
we need to consider two large matrices $\boldsymbol{\mathit{Z}}^{(k)}$ and $\boldsymbol{\mathit{V}}^{(k)}$ with sizes $(n-p)\times p$ and $(n-p)\times 1$, respectively. 
The value $\mathit{y}_{i}^{(k)} - \mathit{x}_{i}$ in $\boldsymbol{\mathit{Z}}^{(k)}$ and $\boldsymbol{\mathit{V}}^{(k)}$ denotes the difference between the labeled/repaired value $\mathit{y}_{i}^{(k)}$ and the input value $\mathit{x}_{i}$ of point $i$.
In practice, the labeled data is often limited, while the repaired data should not be significantly changed referring to the minimum change principle of repairing. 
That is, most values in $\boldsymbol{\mathit{Z}}^{(k)}$ and $\boldsymbol{\mathit{V}}^{(k)}$ equal to 0. 
We show (in Proposition \ref{proposition:matrix-reduce} below) that the sparse matrices could be pruned by removing rows with value 0. 

Let $\mathit{z}^{(k)}_i = \mathit{y}^{(k)}_i - \mathit{x}_i$ for simplicity.  
We rewrite $\boldsymbol{\mathit{V}}^{(k)},\boldsymbol{\mathit{Z}}^{(k)}$ in Equation \ref{equation:ols} for parameter estimation, 
\begin{align*}
\boldsymbol{\mathit{V}}^{(k)} = & 
\begin{pmatrix}
\mathit{z}_{\mathit{p} + 1}^{(k)} \\
\mathit{z}_{\mathit{p} + 2}^{(k)} \\
\vdots \\
\mathit{z}_{n}^{(k)}
\end{pmatrix}
, \quad
\boldsymbol{\mathit{Z}}^{(k)} =  
\begin{pmatrix}
\mathit{z}_{\mathit{p}}^{(k)} & \ldots & \mathit{z}_1^{(k)} \\
\mathit{z}_{\mathit{p}+1}^{(k)} & \ldots & \mathit{z}_2^{(k)} \\
\vdots & \ddots & \vdots \\
\mathit{z}_{n-1}^{(k)} & \ldots & \mathit{z}_{n-\mathit{p}}^{(k)}
\end{pmatrix}
.
\end{align*}

The following conclusion states that 
the same parameter $\phi^{(k)}$ could still be computed by Equation \ref{equation:ols}, after removing the rows in $\boldsymbol{\mathit{Z}}^{(k)}$, whose values equal to 0, and the corresponding rows in $\boldsymbol{\mathit{V}}^{(k)}$.

\begin{proposition}
\label{proposition:matrix-reduce}
For any row in $\boldsymbol{\mathit{Z}}^{(k)}$, denoted by $\boldsymbol{\mathit{Z}}_{r}^{(k)}$, 
if the entire row are all with value 0, 
i.e., $\mathit{z}_{r+\mathit{p}-1}^{(k)}=\mathit{z}_{r+\mathit{p}-2}^{(k)}=\ldots=\mathit{z}_{r}^{(k)}=0$,
then it is safe to remove the row $\boldsymbol{\mathit{Z}}_{r}^{(k)}$ and the corresponding row 
$\boldsymbol{\mathit{V}}_{r}^{(k)} = \begin{pmatrix}\mathit{z}_{\mathit{p}+r}^{(k)}\end{pmatrix}$ 
from matrices $\boldsymbol{\mathit{Z}}^{(k)}$ and $\boldsymbol{\mathit{V}}^{(k)}$,  respectively, 
which still compute that the same $\phi^{(k)}$.
\end{proposition}

\begin{proof}
We rewrite Equation \ref{equation:ols} as
$\phi^{(k)} = (\boldsymbol{\mathit{A}}^{(k)})^{-1}\boldsymbol{\mathit{B}}^{(k)}$, 
where 
$\boldsymbol{\mathit{A}}^{(k)} = (\boldsymbol{\mathit{Z}}^{(k)})'\boldsymbol{\mathit{Z}}^{(k)}$ 
and 
$\boldsymbol{\mathit{B}}^{(k)} = (\boldsymbol{\mathit{Z}}^{(k)})'\boldsymbol{\mathit{V}}^{(k)}$
as illustrated in Equations 
\ref{equation:mid-matrix-A}
and 
\ref{equation:mid-matrix-B},
respectively. 
To prove the same parameter $\phi^{(k)}$, 
we show that all the values in 
$\boldsymbol{\mathit{A}}^{(k)}$ 
and 
$\boldsymbol{\mathit{B}}^{(k)}$ 
are unchanged after removing the $r$-th row with value~0.

(1) For 
$a^{(k)}_{ii}, 1\leq i\leq p$ in Equation \ref{definition:aii}, 
we have \\
$a^{(k)}_{ii}
= \sum\limits_{l=\mathit{p}+1-i}^{n-i}\mathit{z}_{l}^{(k)}\mathit{z}_{l}^{(k)}
=\sum\limits_{l\neq r+\mathit{p}-i}\mathit{z}_{l}^{(k)}\mathit{z}_{l}^{(k)}
+ \mathit{z}_{r+\mathit{p}-i}^{(k)}\mathit{z}_{r+\mathit{p}-i}^{(k)}
= \sum\limits_{l\neq r+\mathit{p}-i}\mathit{z}_{l}^{(k)}\mathit{z}_{l}^{(k)} + 0$.
That is, $a^{(k)}_{ii}$ will not be affected without considering 
$\mathit{z}_{r+\mathit{p}-i}^{(k)}$ in the $r$-th row  $\boldsymbol{\mathit{Z}}_{r}^{(k)}$.

(2) For 
$a^{(k)}_{ij}, a^{(k)}_{ji} ,
1\leq i\leq p, 1\leq j\leq p$ in Equation \ref{definition:aij}, 
we have
$a^{(k)}_{ij} = a^{(k)}_{ji} 
=\sum\limits_{l=\mathit{p}+1-i}^{n-i}\mathit{z}_{l}^{(k)}\mathit{z}_{l-j+i}^{(k)}
=\sum\limits_{l\neq r+\mathit{p}-i}\mathit{z}_{l}^{(k)}\mathit{z}_{l-j+i}^{(k)}
+ \mathit{z}_{r+\mathit{p}-i}^{(k)}\mathit{z}_{r+\mathit{p}-j}^{(k)}
=\sum\limits_{l\neq r+\mathit{p}-i}\mathit{z}_{l}^{(k)}\mathit{z}_{l-j+i}^{(k)} + 0$.
Again, $a^{(k)}_{ij}, a^{(k)}_{ji}$ will not be affected without considering 
$\mathit{z}_{r+\mathit{p}-i}^{(k)}, \mathit{z}_{r+\mathit{p}-j}^{(k)}$ in the $r$-th row  $\boldsymbol{\mathit{Z}}_{r}^{(k)}$.

(3) For
$b^{(k)}_{i}  ,
1\leq i\leq p$ in Equation \ref{definition:bi},  
we have 
$b^{(k)}_{i} 
= \sum\limits_{l=\mathit{p}+1}^{n}\mathit{z}_{l}^{(k)}\mathit{z}_{l-i}^{(k)}
= \sum\limits_{l\neq r+\mathit{p}}\mathit{z}_{l}^{(k)}\mathit{z}_{l-i}^{(k)}
+ \mathit{z}_{r+\mathit{p}}^{(k)}\mathit{z}_{r+\mathit{p}-i}^{(k)}
= \sum\limits_{l\neq r+\mathit{p}}\mathit{z}_{l}^{(k)}\mathit{z}_{l-i}^{(k)} + 0$.
Finally, $b^{(k)}_{i}$ will not be affected without considering 
$\mathit{z}_{r+\mathit{p}}^{(k)}, \mathit{z}_{r+\mathit{p}-i}^{(k)}$ in the $r$-th row  $\boldsymbol{\mathit{Z}}_{r}^{(k)}$ as well.
\end{proof}

\begin{example}[Parameter estimation with matrix pruning, Example \ref{example:parameter-estimation} continued]\label{example:parameter-estimation-mr}
Consider again $\mathit{x} = \{6, 10, 9.6, 8.3, \allowbreak 7.7,  
5.4, 5.6, 5.9, 6.3, 6.8, 7.5, 8.5\}$ and $\mathit{y}^{(0)} = \{6, 5.6, 5.4, 8.3, \allowbreak 7.7, 5.4, \allowbreak 5.6, 5.9, 6.3, 6.8, 7.5, 8.5\}$ 
in Example \ref{example:parameter-estimation}. 
Given order $\mathit{p}=1$, we have $\boldsymbol{\mathit{V}}^{(0)} = \{-4.4,-4.2,0,0,0,\allowbreak 0,0,0,0,0,0\}'$ 
with size $11\times1$,
and $\boldsymbol{\mathit{Z}}^{(0)} = \{0,-4.4,-4.2,0,0,0,\allowbreak 0,0,0,0,0\}'$
with size $11\times1$.
Referring to Proposition \ref{proposition:matrix-reduce}, all the rows except the fourth and fifth rows in $\boldsymbol{\mathit{Z}}^{(0)}$ and the corresponding rows in $\boldsymbol{\mathit{V}}^{(0)}$ can be removed.
After matrix pruning, 
we have $\boldsymbol{\mathit{V}}^{(0)} = \{-4.2,0\}'$ 
with size $2\times1$, and 
$\boldsymbol{\mathit{Z}}^{(0)} = \{-4.4,-4.2\}'$ 
with size $2\times1$.
Referring to Equation \ref{equation:ols}, the parameter is estimated by
$$\phi_1^{(0)} = \frac{(-4.4)*(-4.2)}{(-4.4)^2+(-4.2)^2} = 0.5.$$
The computed parameter is the same as in Example \ref{example:parameter-estimation}.
\end{example}

\subsection{Incremental Computation}
\label{sect:incremental-computation}

\paragraph*{Intuition}
In each iteration in Algorithm \ref{algorithm:arx-iteration}, 
the parameter $\phi^{(k)}$ is estimated by Equation \ref{equation:ols} 
w.r.t.\ $\boldsymbol{\mathit{Z}}^{(k)}$ and $\boldsymbol{\mathit{V}}^{(k)}$
over all the $n$ points. 
However, referring to the minimum change principle in Section \ref{sect:repair-evaluation}, 
there is only one point, say $r$, which is changed in each iteration, 
having $\mathit{y}_{r}^{(k)} \neq \mathit{y}_{r}^{(k-1)}$. 
That is, most values in $\boldsymbol{\mathit{Z}}^{(k)}$ and $\boldsymbol{\mathit{V}}^{(k)}$ are the same to $\boldsymbol{\mathit{Z}}^{(k-1)}$ and $\boldsymbol{\mathit{V}}^{(k-1)}$. 
We show (in Proposition \ref{the:incremental} below) that 
$\phi^{(k)}$ can be incrementally computed by considering only the changed values 
instead of the entire $\boldsymbol{\mathit{Z}}^{(k)}$ and $\boldsymbol{\mathit{V}}^{(k)}$.
The time complexity of parameter estimation in each iteration is thus reduced from 
linear
time to constant time.

\subsubsection{Recursive Formula}

To enable the incremental computation, we rewrite Equation \ref{equation:ols} for parameter estimation as follows, 
\begin{align}
\label{equation:mid-matrix-phi}
\phi^{(k)} &= (\boldsymbol{\mathit{A}}^{(k)})^{-1}\boldsymbol{\mathit{B}}^{(k)}
\\
\label{equation:mid-matrix-A}
\boldsymbol{\mathit{A}}^{(k)} &= (\boldsymbol{\mathit{Z}}^{(k)})'\boldsymbol{\mathit{Z}}^{(k)} = 
 \begin{pmatrix}
  a_{11}^{(k)} & a_{12}^{(k)} & \cdots & a_{1p}^{(k)} \\
  a_{21}^{(k)} & a_{22}^{(k)} & \cdots & a_{2p}^{(k)} \\
  \vdots  & \vdots  & \ddots & \vdots  \\
  a_{p1}^{(k)} & a_{p2}^{(k)} & \cdots & a_{pp}^{(k)} 
 \end{pmatrix}
\\ 
\label{equation:mid-matrix-B}
\boldsymbol{\mathit{B}}^{(k)} &= (\boldsymbol{\mathit{Z}}^{(k)})'\boldsymbol{\mathit{V}}^{(k)}= 
 \begin{pmatrix}
  b_{1}^{(k)}  \\
  b_{2}^{(k)}  \\
  \vdots   \\
  b_{p}^{(k)} 
 \end{pmatrix}
\end{align}
where 
\begin{align}
\label{definition:aii}
a^{(k)}_{ii} &= \sum\limits_{l=\mathit{p}+1-i}^{n-i}\mathit{z}_{l}^{(k)}\mathit{z}_{l}^{(k)} ,
& & 1\leq i\leq p \\ 
\label{definition:aij}
a^{(k)}_{ij} &= a^{(k)}_{ji} =\sum\limits_{l=\mathit{p}+1-i}^{n-i}\mathit{z}_{l}^{(k)}\mathit{z}_{l-j+i}^{(k)},
& & 1\leq i\leq p, 1\leq j\leq p , j>i \\
\label{definition:bi}
b^{(k)}_{i} &= \sum\limits_{l=\mathit{p}+1}^{n}\mathit{z}_{l}^{(k)}\mathit{z}_{l-i}^{(k)} ,
& & 1\leq i\leq p 
\end{align}

The following conclusion illustrates that 
$\boldsymbol{\mathit{A}}^{(k)}$ and $\boldsymbol{\mathit{B}}^{(k)}$, 
with sizes $p\times p$ and $p\times 1$, respectively, could be recursively computed from 
$\boldsymbol{\mathit{A}}^{(k-1)}$ and $\boldsymbol{\mathit{B}}^{(k-1)}$. 
The pre-defined $p$ in the AR($p$)/ARX($p$) model is a fixed value in the algorithm and 
has $p \ll n$. 
The computational cost of Equation \ref{equation:mid-matrix-phi} over $\boldsymbol{\mathit{A}}^{(k)}$ and $\boldsymbol{\mathit{B}}^{(k)}$ with sizes on $p$ 
will be significantly lower than 
Equation \ref{equation:ols} 
w.r.t.\ $\boldsymbol{\mathit{Z}}^{(k)}$ and $\boldsymbol{\mathit{V}}^{(k)}$
over all the $n$ points.

\begin{proposition}
\label{the:incremental}
Let $r$ be the changed point in the current repair iteration, having 
$\mathit{y}_{r}^{(k)} \neq \mathit{y}_{r}^{(k-1)}$ or equivalently 
$\mathit{z}_{r}^{(k)} \neq \mathit{z}_{r}^{(k-1)}$. 
$\boldsymbol{\mathit{A}}^{(k)}, \boldsymbol{\mathit{B}}^{(k)}$ could be recursively computed from 
$\boldsymbol{\mathit{A}}^{(k-1)}, \boldsymbol{\mathit{B}}^{(k-1)}$. 
\end{proposition}
\noindent That is, for $1\leq i \leq p$, we have 
\begin{align}
\label{equation:estimate-recursive-a}
a_{ii}^{(k)} &= 
a_{ii}^{(k-1)} + \\ \notag
&
\begin{cases}
0 
& \quad \textrm{if } r < \mathit{p}+1-i \textrm{ or } r > n-i \\
\mathit{z}_r^{(k)}\mathit{z}_r^{(k)} -\mathit{z}_r^{(k-1)}\mathit{z}_r^{(k-1)}
& \quad \textrm{if }  \mathit{p}+1-i \leq r \leq n-i
\end{cases}  
\end{align}
For $1\leq i \leq p, 1\leq j \leq p, i<j$, we have 
\begin{align} \label{equation:estimate-recursive-b}
a_{ij}^{(k)} &= a_{ji}^{(k)}
=a_{ij}^{(k-1)} + (\mathit{z}_r^{(k)}-\mathit{z}_r^{(k-1)})\times \\ \notag
&
\begin{cases}
0 
& \quad \textrm{if } r < \mathit{p}+1-j \textrm{ or } r>n-i \\
\mathit{z}_{r+j-i}^{(k-1)}
& \quad \textrm{if } p+1-j \leq r < p+1-i \\
\mathit{z}_{r-j+i}^{(k-1)}
& \quad \textrm{if } n-j < r \leq n-i \\
(\mathit{z}_{r+j-i}^{(k-1)}+\mathit{z}_{r-j+i}^{(k-1)})
& \quad \textrm{if } p+1-i \leq r \leq n-j
\end{cases}  
\end{align}
For $1\leq i \leq p$, we have 
\begin{align} \label{equation:estimate-recursive-c}
b_{i}^{(k)} &
=b_{i}^{(k-1)} + (\mathit{z}_r^{(k)}-\mathit{z}_r^{(k-1)})\times \\ \notag
& 
\begin{cases}
0 
& \quad \textrm{if } r<\mathit{p}+1-i \\
\mathit{z}_{r+i}^{(k-1)}
& \quad \textrm{if } p+1-i \leq r < p+1 \\
\mathit{z}_{r-i}^{(k-1)}
& \quad \textrm{if } r>n-i \\
(\mathit{z}_{r+i}^{(k-1)}+\mathit{z}_{r-i}^{(k-1)})
& \quad \textrm{if } p+1 \leq r \leq n-i
\end{cases} 
\end{align}

\begin{proof}
\begin{figure}[h]
\centering
\includegraphics[width=0.85\figwidths]{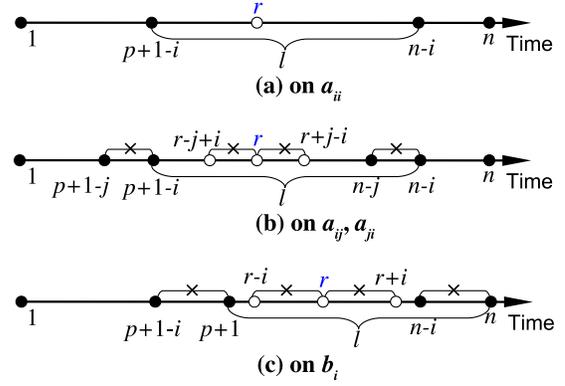}
\caption{Updates for changed point $r$}
\label{fig:imric}
\end{figure}

We show correctness of Equations \ref{equation:estimate-recursive-a}-\ref{equation:estimate-recursive-c} as follows.

(1) For $a_{ii}^{(k)}$ in Equation \ref{definition:aii}, 
it is the summation of $\mathit{z}_{l}^{(k)}$ values for data points $l$ starting from $p+1-i$ to $n-i$.
Recall that there is only one repaired point $r$ in each iteration 
with value changed from $\mathit{z}_r^{(k-1)}$ to $\mathit{z}_r^{(k)}$. 
If $r$ is in the range of summation, i.e., 
$\mathit{p}+1-i \leq r \leq n-i$ as illustrated in Figure \ref{fig:imric}(a), 
we need an update on $a_{ii}^{(k-1)}$ for $a_{ii}^{(k)}$ by adding 
$\mathit{z}_r^{(k)}\mathit{z}_r^{(k)}$ and removing $\mathit{z}_r^{(k-1)}\mathit{z}_r^{(k-1)}$. 
Since other points are not modified, the correctness of Equation \ref{equation:estimate-recursive-a} is illustrated. 

(2) For $a_{ij}^{(k)}$ in Equation \ref{definition:aij}, 
the changed point $r$ may affect two terms,  
$\mathit{z}_{r}^{(k)}\mathit{z}_{r-j+i}^{(k)}$
and 
$\mathit{z}_{r+j-i}^{(k)}\mathit{z}_{r}^{(k)}$, 
in the summation with $l=r$ and $r+j-i$, respectively. 
A term $\mathit{z}_{l}^{(k)}\mathit{z}_{l-j+i}^{(k)}$ is represented by `$\times$' in Figure \ref{fig:imric}(b).
There are four cases:
\\
i) If both terms are not in the summation range, i.e., 
$r<\mathit{p}+1-j$ or $r>n-i$, 
$a_{ij}^{(k)}$ will not change.
\\
ii) For $\mathit{p}+1-i\leq r < \mathit{p}+1$, 
the term $\mathit{z}_{r+j-i}^{(k)}\mathit{z}_{r}^{(k)}$ will be under consideration in the summation. 
We need an update on $a_{ij}^{(k-1)}$ for $a_{ij}^{(k)}$ by adding 
$\mathit{z}_{r+j-i}^{(k)}\mathit{z}_{r}^{(k)}$ and removing 
$\mathit{z}_{r+j-i}^{(k-1)}\mathit{z}_{r}^{(k-1)}$. 
Since point $r+j-i$ is not changed, we have 
$\mathit{z}_{r+j-i}^{(k)}=\mathit{z}_{r+j-i}^{(k-1)}$.
The correctness of the second case in Equation \ref{equation:estimate-recursive-b} is illustrated. 
\\
iii) Similarly, for $n-j< r \leq n-i$, 
the term $\mathit{z}_{r}^{(k)}\mathit{z}_{r-j+i}^{(k)}$ is under consideration. 
An update 
$(\mathit{z}_r^{(k)}-\mathit{z}_r^{(k-1)})\times 
\mathit{z}_{r-j+i}^{(k-1)}$
needs to be applied on $a_{ij}^{(k-1)}$ for $a_{ij}^{(k)}$.
\\
iv) If both terms are in the range, i.e., 
$p+1-i \leq r \leq n-j$, 
we introduce an update on $a_{ij}^{(k-1)}$ for $a_{ij}^{(k)}$ by adding 
$(\mathit{z}_{r}^{(k)}-\mathit{z}_{r}^{(k-1)})\times(\mathit{z}_{r+j-i}^{(k-1)}+\mathit{z}_{r-j+i}^{(k-1)})$.

(3) For $b_{i}^{(k)}$ in Equation \ref{definition:bi}, 
similar to $a_{ij}^{(k)}$, 
the changed point $r$ may affect two terms,  
$\mathit{z}_{r}^{(k)}\mathit{z}_{r-i}^{(k)}$
and 
$\mathit{z}_{r+i}^{(k)}\mathit{z}_{r}^{(k)}$, 
in the summation with $l=r$ and $r+i$, respectively. 
There are four cases to consider as well, as illustrated in Figure \ref{fig:imric}(c).
Similar proofs in (2) for each case apply.   
\end{proof}

\subsubsection{Recursive Algorithm}

Algorithm \ref{algorithm:repair-parameter} shows the procedure of incrementally estimating the parameter $\phi^{(k)}$.
For $k=0$ in the first iteration, computing $\boldsymbol{\mathit{A}}^{(0)}, \boldsymbol{\mathit{B}}^{(0)}$ by Equations \ref{equation:mid-matrix-A} and \ref{equation:mid-matrix-B} w.r.t.\  matrices $\boldsymbol{\mathit{Z}}^{(0)},\boldsymbol{\mathit{V}}^{(0)}$, however, is inevitable. 
Nevertheless, the efficient pruning of rows with value 0 in Proposition \ref{proposition:matrix-reduce} can be applied 
as presented in Line \ref{algorithm:estimate-prune} in Algorithm \ref{algorithm:repair-parameter}. 

\begin{algorithm}[h]
\label{algorithm:repair-parameter}
\caption{$\mathsf{Estimate}(\mathit{x}, \mathit{y}^{(k)})$}
 \KwIn{time series $\mathit{x}$ and  intermediate repair result $\mathit{y}^{(k)}$}
 \KwOut{estimated parameter $\phi^{(k)}$}
 \If{$k = 0$}{
 Initialize $\boldsymbol{\mathit{A}}^{(0)}, \boldsymbol{\mathit{B}}^{(0)}$ in Equations \ref{equation:mid-matrix-A} and \ref{equation:mid-matrix-B} by using the pruned matrices $\boldsymbol{\mathit{Z}}^{(0)},\boldsymbol{\mathit{V}}^{(0)}$ in Proposition \ref{proposition:matrix-reduce} \; \label{algorithm:estimate-prune}
 }
 \Else{
 Let $r$ be the changed point in the $k$-th iteration having
 $\mathit{y}_{r}^{(k)} \neq \mathit{y}_{r}^{(k-1)}$\;
 Compute $\boldsymbol{\mathit{A}}^{(k)}, \boldsymbol{\mathit{B}}^{(k)}$ according to 
 $\boldsymbol{\mathit{A}}^{(k-1)}, \boldsymbol{\mathit{B}}^{(k-1)}$ by using the recursive Equations \ref{equation:estimate-recursive-a}-\ref{equation:estimate-recursive-c} \; \label{algorithm:update}
 }
 $\phi^{(k)} \leftarrow (\boldsymbol{\mathit{A}}^{(k)})^{-1}\boldsymbol{\mathit{B}}^{(k)}$ according to Equation \ref{equation:mid-matrix-phi} \;
 \Return{$\phi^{(k)}$}
\end{algorithm}

For the following iterations $k>0$, the recursive computing of $\boldsymbol{\mathit{A}}^{(k)}, \boldsymbol{\mathit{B}}^{(k)}$ from
 $\boldsymbol{\mathit{A}}^{(k-1)}, \boldsymbol{\mathit{B}}^{(k-1)}$ performs. 
As presented in Proposition \ref{the:incremental}, 
all the $p^2+p$ values in $\boldsymbol{\mathit{A}}^{(k)}, \boldsymbol{\mathit{B}}^{(k)}$ can be recursively updated in constant time. 
Consequently, 
the complexity of parameter estimation is reduced from $O(n)$ (referring to Equations \ref{definition:aii}-\ref{definition:bi}) to 
$O(1)$ in Equations \ref{equation:estimate-recursive-a}-\ref{equation:estimate-recursive-c}.

\begin{example}[Parameter estimation using incremental computation, Example \ref{example:parameter-estimation} continued]
\label{example:parameter-estimation-ic}
Consider again $\mathit{x} = \{6, 10, \allowbreak
9.6, 8.3, 7.7, 5.4, 5.6, 5.9, 6.3, 6.8, 7.5, 8.5\}$ and $\mathit{y}^{(0)} = \{6, 5.6, \allowbreak5.4, 8.3, 7.7, 5.4, 5.6, 5.9, 6.3, 6.8, 7.5, 8.5\}$ 
in Example \ref{example:parameter-estimation}.
We have $\boldsymbol{\mathit{V}}^{(0)} = \{-4.4,-4.2,0,0,0,\allowbreak 0,0,0,0,0,0\}'$ 
and $\boldsymbol{\mathit{Z}}^{(0)} = \{0,-4.4,-4.2,0,0,0,\allowbreak 0,0,0,0,0\}'$.
Given $\mathit{p}=1$, the matrices $\boldsymbol{\mathit{A}}^{(0)}, \boldsymbol{\mathit{B}}^{(0)}$ have only one element, 
with $a_{11}^{(0)}=(-4.4)^2+(-4.2)^2=37$ and $b_{1}^{(0)}=(-4.4)*(-4.2)=18.48$
initialized by Line \ref{algorithm:estimate-prune} in Algorithm \ref{algorithm:repair-parameter}.

According to Examples \ref{example:repair-candidate} and \ref{example:repair-evaluation}, 
the repaired point is $\mathit{y}_{4}^{(1)}=6.2$.
We have $\mathit{z}_{4}^{(1)}=6.2-8.3=-2.1$, while $\mathit{z}_{4}^{(0)}=0$.
Line \ref{algorithm:update} in Algorithm \ref{algorithm:repair-parameter} incrementally computes 
$\boldsymbol{\mathit{A}}^{(1)}, \boldsymbol{\mathit{B}}^{(1)}$ from 
$\boldsymbol{\mathit{A}}^{(0)}, \boldsymbol{\mathit{B}}^{(0)}$, 
i.e.,  $a_{11}^{(1)}=a_{11}^{(0)}+(-2.1)^2=41.41$
by using the incremental update in Equation \ref{equation:estimate-recursive-a}, and
$b_{1}^{(1)}=b_{1}^{(0)}+(-2.1-0)*(-4.2+0)=27.3$ referring to Equation \ref{equation:estimate-recursive-c}.
Finally, the parameter ${\phi}_{1}^{(1)}$ is computed according to 
$\boldsymbol{\mathit{A}}^{(1)}, \boldsymbol{\mathit{B}}^{(1)}$
by using Equation \ref{equation:mid-matrix-phi}, 
i.e., ${\phi}_{1}^{(1)}=27.3/41.41=0.66$.
\end{example}

%-------------------------------------------------------------------------

\section{Experiment}\label{sect:experiment}

In this section, we experimentally compare our proposed methods IMR with the state-of-the-art approaches,
including the anomaly detection methods using 
(1) AR \cite{box1994time}, (2) ARX \cite{box1994time}, (3) ARIMA \cite{otto1990two,box1994time},
(4) Tsay \cite{tsay1988outliers} as models, 
(5) the smoothing-based method EWMA \cite{hellerstein2008quantitative}, and 
(6) the constraint-based approach SCREEN \cite{DBLP:conf/sigmod/SongZWY15}.

\paragraph*{GPS data with real errors} 

In the GPS dataset,
real errors are naturally embedded and the corresponding ground truths are manually labeled. 
It collects GPS readings by a person carrying a smartphone and walking around at campus. Since we know exactly the path of walking, a number of 186 dirty points out of 742 points in trajectory are  manually identified. 
True locations of dirty points are also manually labeled, as ground truth. 
(See major results in Section \ref{sect:GPS}.)

\paragraph*{ILD data with synthetic  errors} 

The Intel Lab Data (\url{http://db.csail.mit.edu/labdata/labdata.html},
ILD) 
includes a number of measurements taken from 54 sensors for every 31 seconds in about 38 days. 
Taking 31 seconds as one epoch and omitting the missing data, a dataset of 4912 points is obtained in sensor 1 from Feb 29th to Mar 1st. 
We synthetically inject errors into the data, by shifting the values for an amount of 3 with variance 0.1 under Gaussian distribution (see some examples in Figure \ref{fig:example-ild}). 
Such ``shifting'' errors are very common in practice, for example the sensor is stuck for a short while, or unit error in collection in a period. 
(See major results in Section \ref{sect:ILD}.)

\paragraph*{Criteria}

RMS error \cite{DBLP:conf/vldb/JefferyGF06} is employed to evaluate the repair.
Let $\mathit{x}^\text{truth}$ be the ground truth of clean sequence, $\mathit{x}^\text{dirty}$ be the observation sequence with faults embedded, and $\mathit{x}^\text{repair}$ be the repaired sequence.  
The RMS error \cite{DBLP:conf/vldb/JefferyGF06} is given by: 
$$\Delta(\mathit{x}^\text{truth},\mathit{x}^\text{repair})=
\sqrt{\frac{1}{n}\sum_{i=1}^n(\mathit{x}^\text{truth}_i-\mathit{x}^\text{repair}_i)^2}.$$
The measure evaluates the distance between the ground truth and its repair result. 
Low RMS error is preferred.

\begin{figure}[t]
\centering
\includegraphics[width=\figwidths]{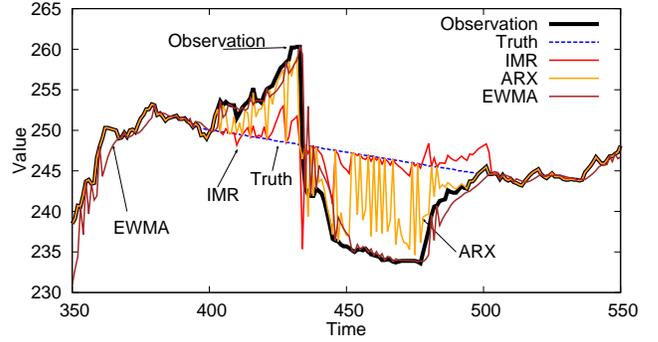}
\caption{GPS example}
\label{fig:example-gps}
\end{figure}

\subsection{Experiments on Real Errors}
\label{sect:GPS}

The experiments on real errors over GPS data consider various algorithm settings, 
including 
(1) order $\mathit{p}$, 
(2) convergence threshold $\tau$, 
(3) max-num-iterations,
and 
(4) labeling rate. 
Similar results are also observed in ILD and omitted.

\subsubsection{Example Results}
\label{sect:exp-example}

Figure \ref{fig:example-gps} illustrates an example part (in latitude, after transformed) of the GPS dataset, including the collected observations with errors, the labeled truth, and the repair results by different methods. 
Owing to various influences such as buildings,  GPS readings may deviate from the truth.  
For instance, the data between time points 400 and 500 are collected from a place near a high building, where significant errors are observed. 
As shown, the proposed IMR shows a repair closest to the truth, compared to other methods.

\begin{figure}[t]
\begin{minipage}{\expwidths}\centering
\hspace{-0.5em}%
\includegraphics[width=\expwidths]{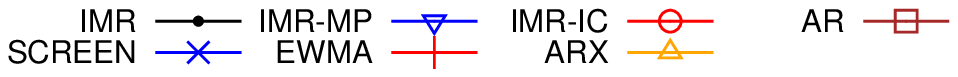}
\hspace{-0.5em}%
\includegraphics[width=0.5\expwidths]{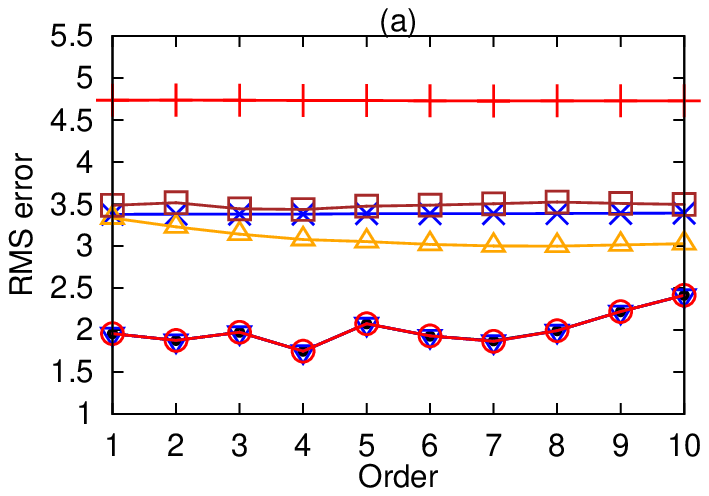}%
\hspace{-0.5em}%
\includegraphics[width=0.5\expwidths]{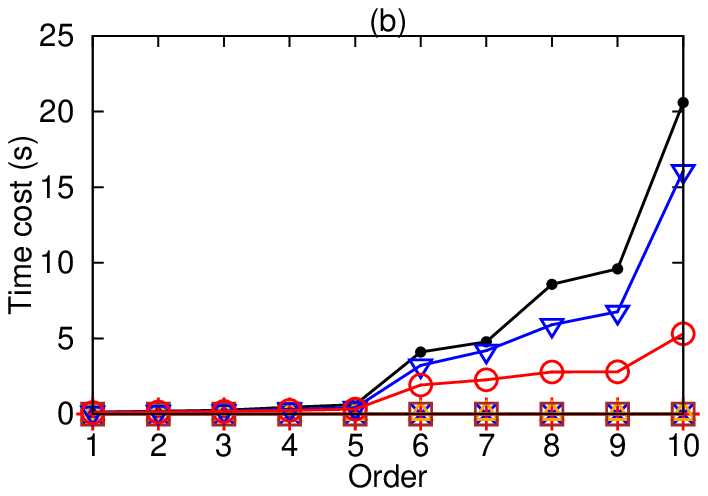}%
\end{minipage}
\caption{Varying order $\mathit{p}$, over GPS with $\tau=0.2$, data size 750, and labeling rate 0.2}
\label{exp:gps-arxp}
\end{figure}

\begin{figure}[t]
\begin{minipage}{\expwidths}\centering
\hspace{-0.5em}%
\includegraphics[width=\expwidths]{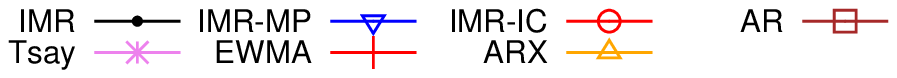} 
\hspace{-0.5em}%
\includegraphics[width=0.5\expwidths]{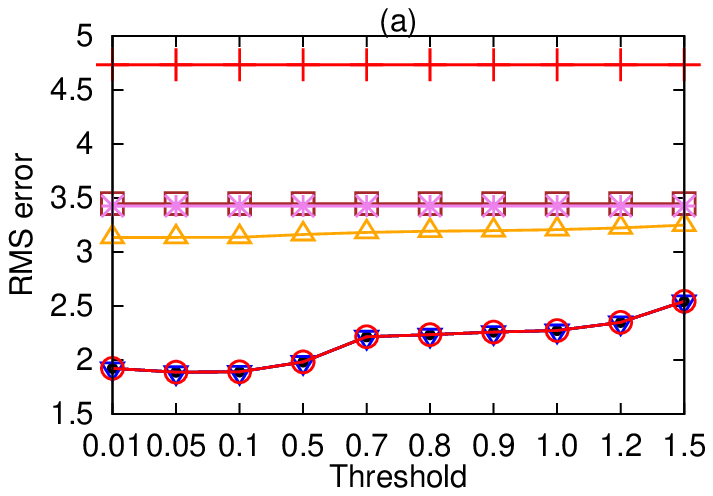}%
\hspace{-0.5em}%
\includegraphics[width=0.5\expwidths]{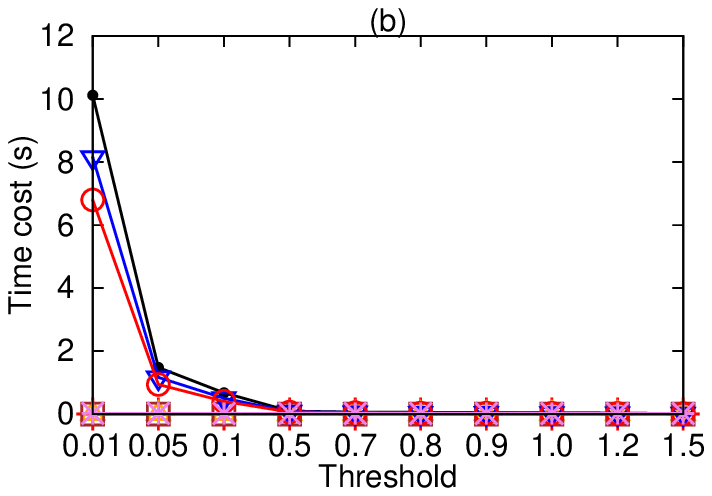}%
\end{minipage}
\caption{Varying threshold $\tau$, over GPS with $p=3$, data size 750, and labeling rate 0.2}
\label{exp:gps-delta}
\end{figure}

\begin{figure}[t]\centering
\begin{minipage}{\expwidths}\centering
\hspace{-0.5em}%
\includegraphics[width=\expwidths]{exp-label-hor-camera} 
\hspace{-0.5em}%
\includegraphics[width=0.5\expwidths]{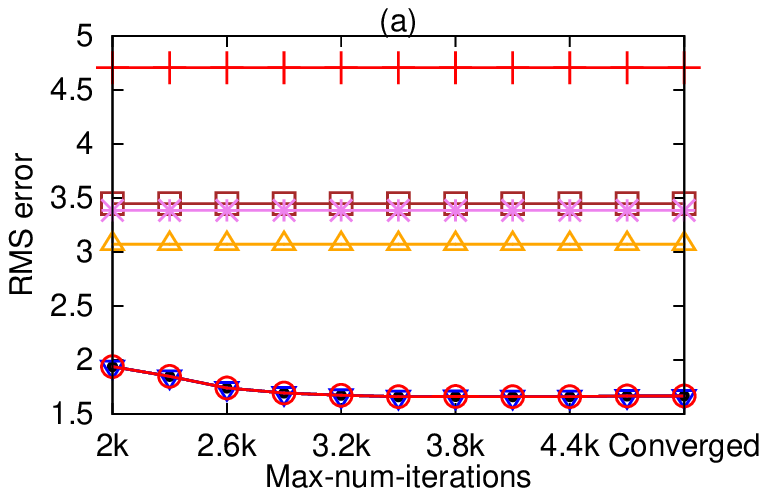}
\hspace{-0.5em}%
\includegraphics[width=0.5\expwidths]{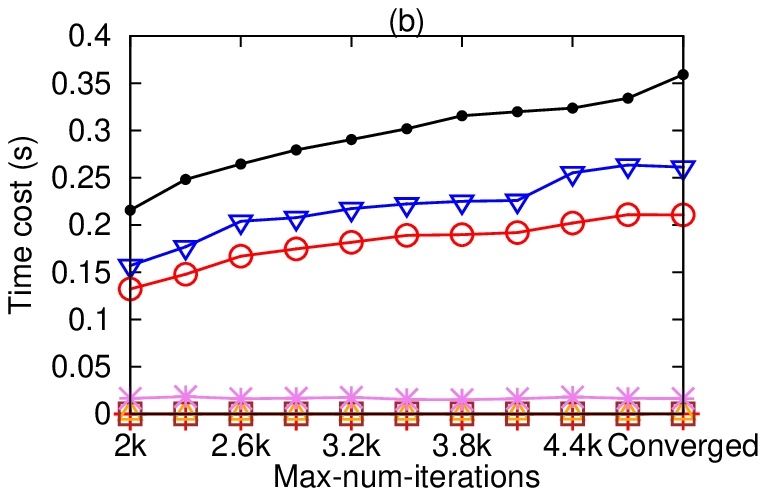}
\end{minipage}
\caption{Varying maximum number of iterations, over GPS with $\tau=0.2, \mathit{p}=3$ and data size 750}
\label{exp:gps-k}
\end{figure}

\begin{figure}[t]
\begin{minipage}{\expwidths}\centering
\hspace{-0.5em}%
\includegraphics[width=\expwidths]{exp-label-hor} 
\hspace{-0.5em}%
\includegraphics[width=0.5\expwidths]{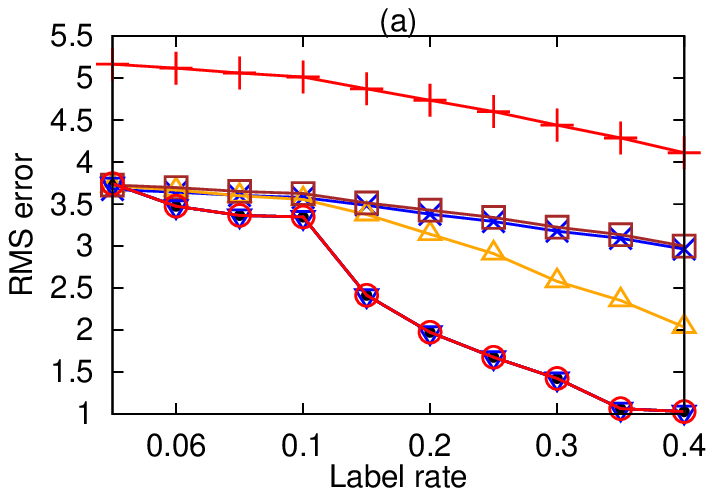}%
\hspace{-0.5em}%
\includegraphics[width=0.5\expwidths]{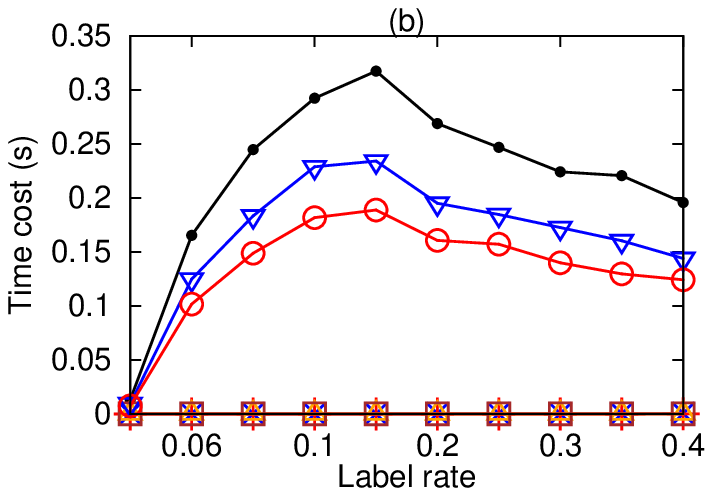}%
\end{minipage}
\caption{Varying labeling rate, over GPS with $\tau=0.2, \mathit{p}=3$ and data size 750}
\label{exp:gps-lrate}
\end{figure}

\begin{figure}[t]
\centering
\includegraphics[width=\figwidths]{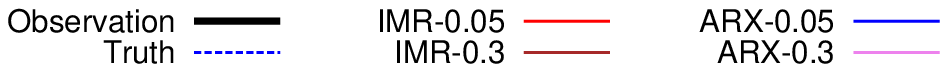}
\includegraphics[width=\figwidths]{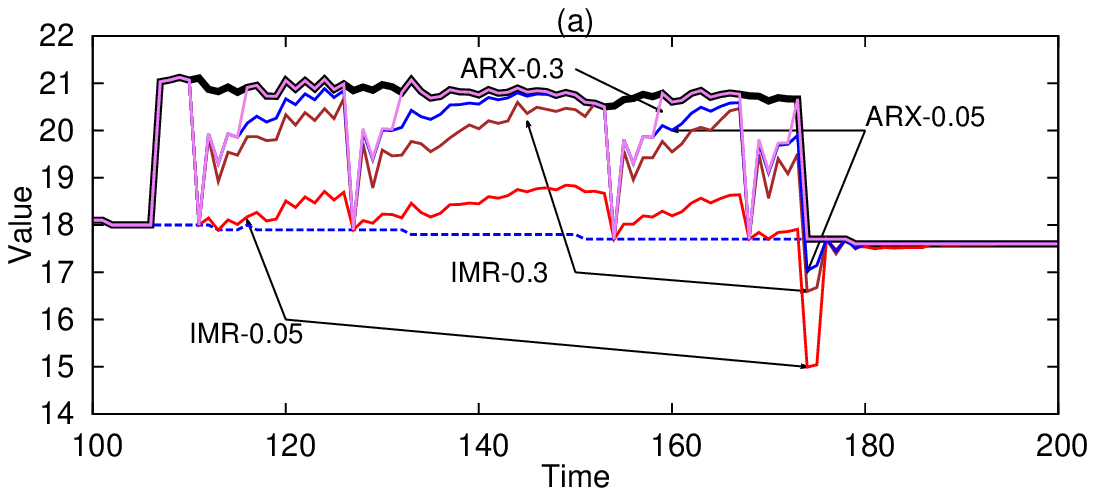} \\
\includegraphics[width=\figwidths]{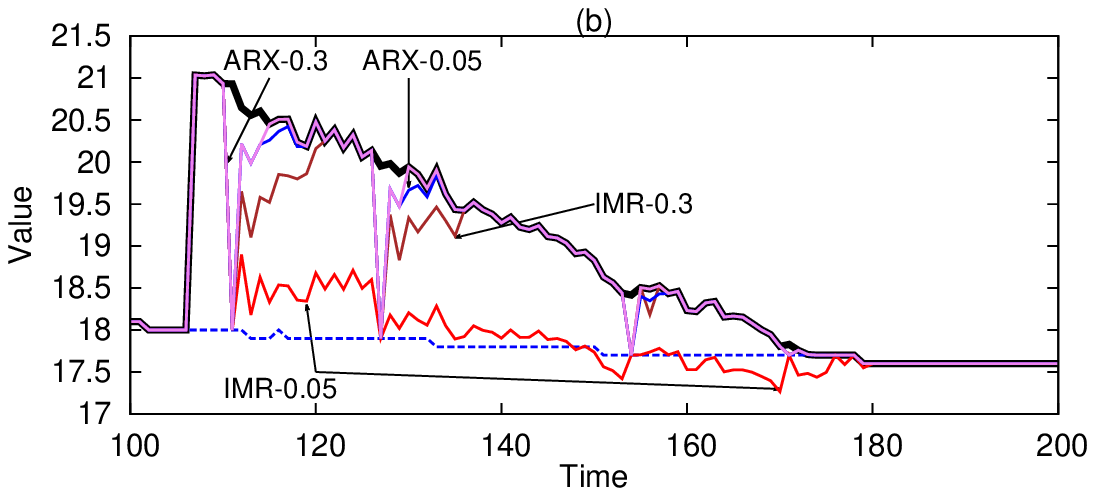}
\caption{ILD example with (a) Shift and (b) Innovational errors  }
\label{fig:example-ild}
\end{figure}

\begin{figure}[t]\centering
\begin{minipage}{\expwidths}\centering
\hspace{-0.5em}%
\includegraphics[width=\expwidths]{exp-label-hor} 
\hspace{-0.5em}%
\includegraphics[width=0.5\expwidths]{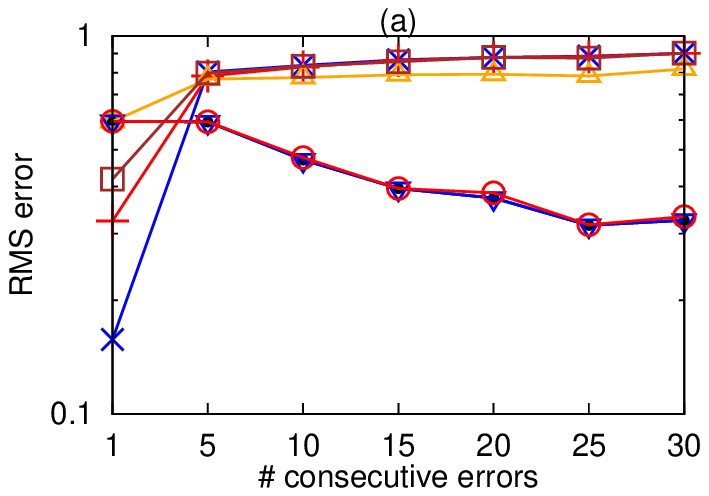}
\hspace{-0.5em}%
\includegraphics[width=0.5\expwidths]{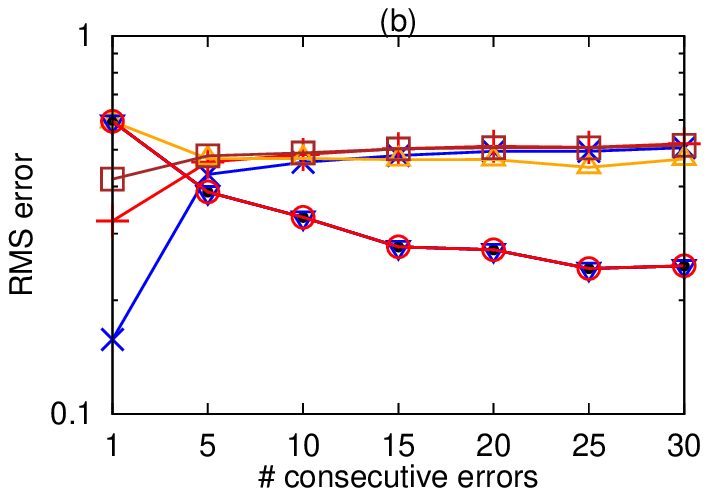} 
\end{minipage}
\caption{Varying the number of consecutive errors, under (a)Shift and (b) Innovational error types, over ILD with $\tau=0.1,\mathit{p}=3$ and data size 3k}
\label{exp:ild-consize}
\end{figure}

\subsubsection{Varying Order $\mathit{p}$}
\label{sect-GPS-order}

Figure \ref{exp:gps-arxp} presents the results on varying order $p$, for 
AR($\mathit{p}$), ARX($\mathit{p}$) and IMR($\mathit{p}$). 
First, as shown in Figure \ref{exp:gps-arxp}(a), AR-based method shows better performance with the increase of order $p$, where more historical values take effect in the predication of each value.
An excessively large $p$, however, does not show further improvement, 
since this simple model may not be able to capture the complicated semantics in a large window.  
Similar results are also observed in ARX for the same reason. 
Remarkably, owing to the iterative strategy with minimum repair in each iteration, our proposed IMR method already achieves low RMS error of repairing even with $p=1$. 
The results verify the necessity of analyzing the special case of IMR(1) with $p=1$ in Section \ref{sect:analysis}.

It is not surprising that the iterative IMR needs higher time costs in Figure \ref{exp:gps-arxp}(b) than other existing methods with only one pass through the data. 
In addition to the original IMR in Algorithm \ref{algorithm:arx-iteration}, 
we also present the results of IMR with matrix pruning (IMR-MP) in Section \ref{sect:matrix-pruning} 
and incremental computation (IMR-IC) in Algorithm \ref{algorithm:repair-parameter} for efficient parameter estimation. 
IMR, IMR-MP and IMR-IC show exactly the same accuracy results in Figure \ref{exp:gps-arxp}(a).
Both efficient estimation methods improve the time costs in Figure \ref{exp:gps-arxp}(b). 
In particular, IMR-IC for incremental parameter estimation with constant time significantly reduces time costs. 

In summary, as illustrated in Figures \ref{exp:gps-arxp} and \ref{exp:ild-p} over the GPS and ILD datasets, respectively, our proposed IMR has no clear preference of order $\mathit{p}$ in repairing accuracy, 
while larger order $\mathit{p}$ leads to higher time cost.

\subsubsection{Varying Convergence Threshold $\tau$}
\label{sect-experiment-threshold}

Figure \ref{exp:gps-delta} reports the results by varying the threshold $\tau$ (with $p=3$). 
By setting a small $\tau$, IMR needs more iterations to converge.
The corresponding time costs in Figure \ref{exp:gps-delta}(b) are higher. 
Better repairing performance is achieved by IMR with a small $\tau$, as shown in Figure \ref{exp:gps-delta}(a).
However, by further reducing the threshold $\tau$, e.g., from 0.1 to 0.01, the repair accuracy could hardly be further improved, while the corresponding iterations and time costs significantly increase. 
On the other hand, by increasing the threshold $\tau$, the time costs reduce. 
Indeed,
the threshold $\tau$ provides a trade-off between repair accuracy and time costs for IMR. 

In summary, a lower threshold indeed leads to better results (lower RMS error) and needs more iterations (higher time costs). 
See Figure \ref{exp:ild-delta} on ILD for more clear impact of the threshold.

\subsubsection{Specifying Maximum Number of Iterations}
\label{sect:max-iteration-exp}

In Section \ref{sect:analysis}, we analyze several special cases, 
where repairing is guaranteed to converge in theory under some conditions.  
For general cases where such conditions are not met, 
(although all the experiments converge in Section \ref{sect:experiment} 
under various settings with/without theoretical convergence guarantee),  
one may specify the maximum number of iterations, as a remedy in practice to avoid waiting for convergence.
That is, Algorithm \ref{algorithm:arx-iteration} terminates when the iteration number reaches \emph{max-num-iterations}, even if the convergence condition in Line \ref{algorithm:repair-converge} is not met. 

Figures 
\ref{exp:gps-k} evaluates various settings of maximum number of iterations 
(average time costs are reported by repeating each test 10 times).%
As illustrated, a moderately large number of iterations already achieve good repair results, i.e., close to the (right-most) converged results.

\subsubsection{Varying Labeling Rate}

Figure \ref{exp:gps-lrate} illustrates the results on various labeling rates. 
A labeling rate 0.1 denotes that 10\% data points are labeled with truth in the dataset. 
It is not surprising that the higher the labeling rate is,  the better the repair performance of IMR and ARX will be, which utilize the labeled truth, 
as illustrated in Figure \ref{exp:gps-lrate}(a). 
An interesting result is that with the increase of labeling rate, 
the corresponding time costs in Figure \ref{exp:gps-lrate}(b), first increase and then drop. 
The reason is that for a small labeling rate (say 0.06) with points barely modified, 
the iterative repair can quickly converge, while leaving most dirty data unchanged. 
The corresponding RMS error is high in this case as shown in Figure \ref{exp:gps-lrate}(a). 
With more data labeled in the input, more dirty points will be identified and repaired by the algorithm, leading to higher computation costs. 
When the labeling rate is large, such as 0.25, a great number of dirty points may be labeled. Thereby, the iterative repair could  converge quickly again. 

As anomaly detection methods, the results of Tsay and ARIMA are generally similar to those of AR and ARX. 
SCREEN and EWMA methods are not affected by  order $p$, threshold $\tau$ and  labeling rate. 
It is not surprising that SCREEN performs weakly, which verifies our motivation and analysis in the Introduction. 
Similarly, since EWMA does not utilize the labeled truth, its performance is weak.

\begin{figure}[t]
\begin{minipage}{\expwidths}\centering
\hspace{-0.5em}%
\includegraphics[width=\expwidths]{exp-label-hor}
\hspace{-0.5em}%
\includegraphics[width=0.5\expwidths]{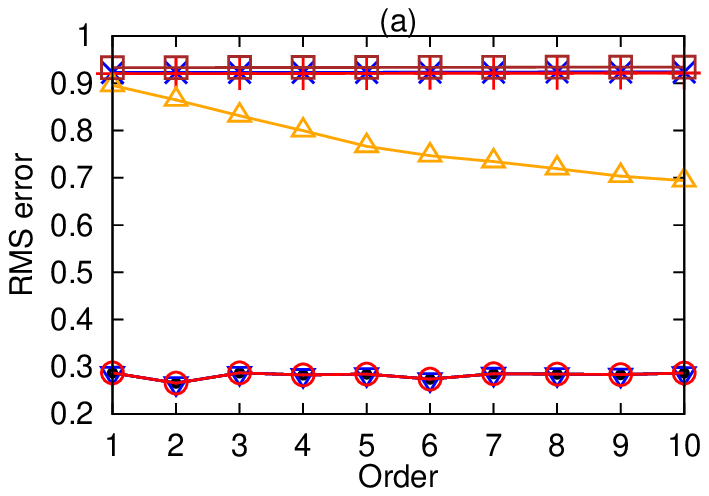}%
\hspace{-0.5em}%
\includegraphics[width=0.5\expwidths]{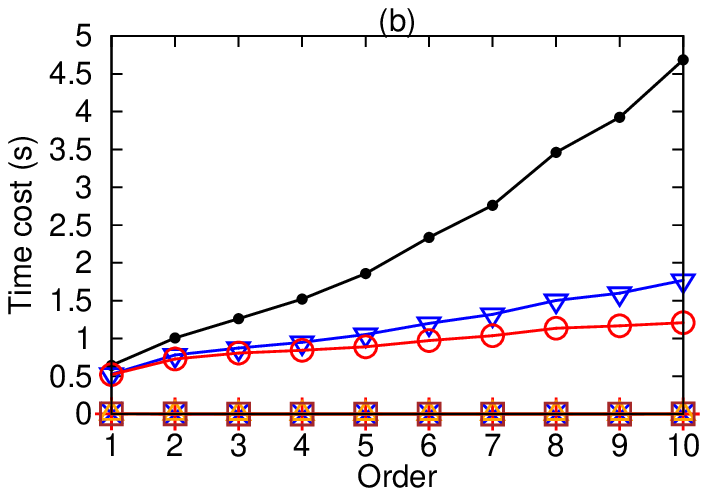}%
\end{minipage}
\caption{Varying order $\mathit{p}$, over ILD with $\tau=0.1$, data size 3k, and labeling rate 0.2}
\label{exp:ild-p}
\end{figure}

\begin{figure}[t]
\begin{minipage}{\expwidths}\centering
\hspace{-0.5em}%
\includegraphics[width=\expwidths]{exp-label-hor-camera} 
\hspace{-0.5em}%
\includegraphics[width=0.5\expwidths]{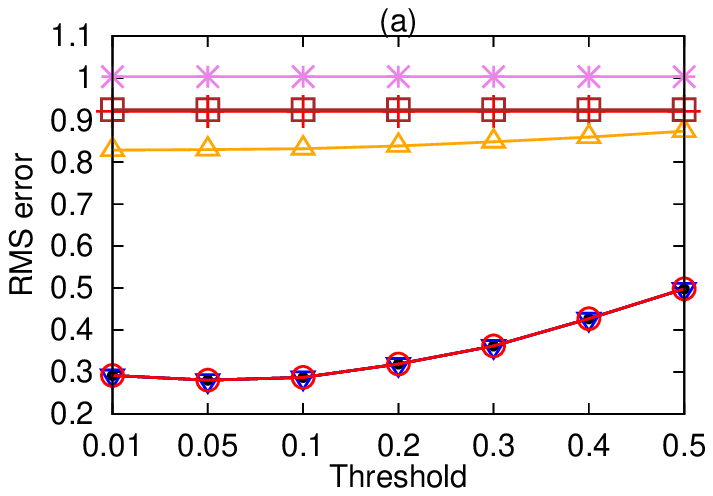}%
\hspace{-0.5em}%
\includegraphics[width=0.5\expwidths]{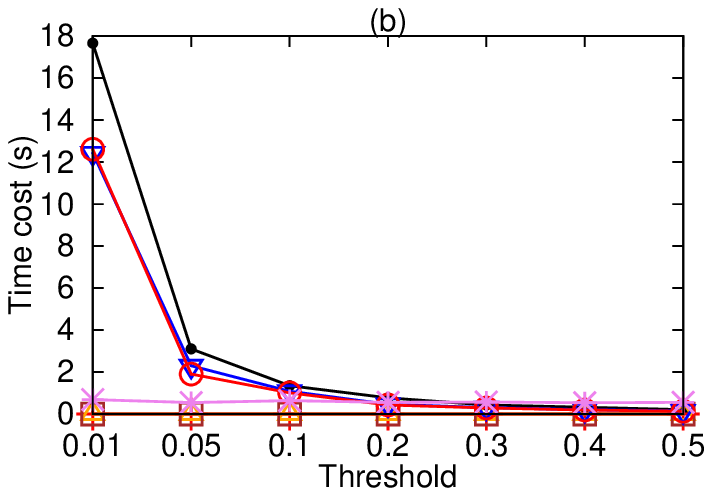}%
\end{minipage}
\caption{Varying threshold $\tau$, over ILD with $p=3$, data size 3k, and labeling rate 0.2}
\label{exp:ild-delta}
\end{figure}

\subsection{Experiments on Synthetic Errors}
\label{sect:ILD}

The experiments on synthetic errors over ILD data focus on varying the errors.
Again, similar results are also observed in the other dataset GPS and thus omitted.

\subsubsection{Evaluation on Various Errors}
\label{sect:exp-ILD-error}

We consider Shift and Innovational errors \cite{box1975intervention, tsay1988outliers}, as the example illustrated in Figure \ref{fig:example-ild}, where Spike errors are considered as a special case with \# consecutive errors = 1. 
The corresponding accuracy results are reported in  Figure \ref{exp:ild-consize}. 
In general, similar results are observed over Innovational and Shift errors. 
That is, while our proposed IMR may not deal with Spike errors (i.e., with \# consecutive errors = 1 in Figure \ref{exp:ild-consize}) as good as SCREEN \cite{DBLP:conf/sigmod/SongZWY15}, 
IMR always shows significantly better results (lower RMS measure) in addressing a large number of consecutive errors, on both Shift and Innovational error patterns. 
The results demonstrate again that 
our proposal works well in repairing consecutive errors.

As illustrated in Figure \ref{fig:example-ild}(a), not only the proposed IMR but also ARX with a small threshold $\tau$ suffers from an overcorrection when the time series shifts back to non-anomalous data. 
The number attached to each method, e.g., IMR-0.05, denotes $\tau=0.05$ for IMR.
As shown, there is a trade-off in both IMR and ARX: 
a smaller threshold $\tau$ shows better results in dealing with consecutive errors, but leads to overcorrection when the shift ends.  
Nevertheless, as presented in Figures \ref{exp:gps-delta} and \ref{exp:ild-delta}, 
a smaller threshold $\tau$ generally has better overall accuracy.
It is also worth noting that existing methods, EWMA smoothing and SCREEN, cannot handle well the repairing either when the time series shifts back to non-anomalous data, as the example illustrated in Figure \ref{fig:simple-example}.

\begin{figure}[t]\centering
\begin{minipage}{\expwidths}\centering
\hspace{-0.5em}%
\includegraphics[width=\expwidths]{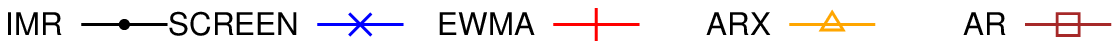} 
\hspace{-0.5em}%
\includegraphics[width=0.5\expwidths]{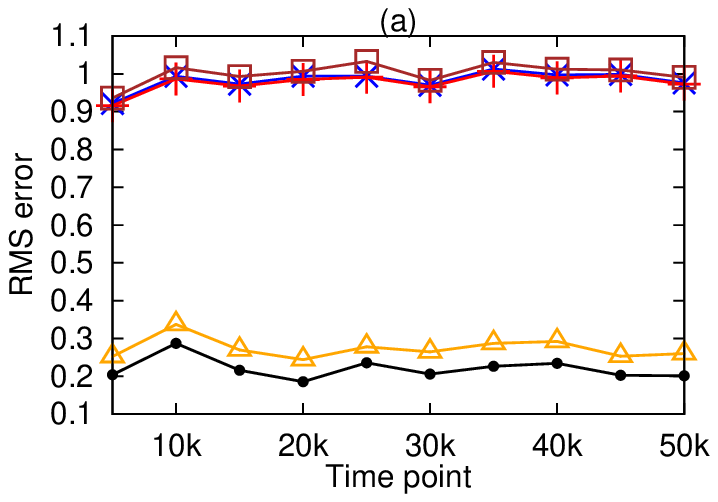}
\hspace{-0.5em}%
\includegraphics[width=0.5\expwidths]{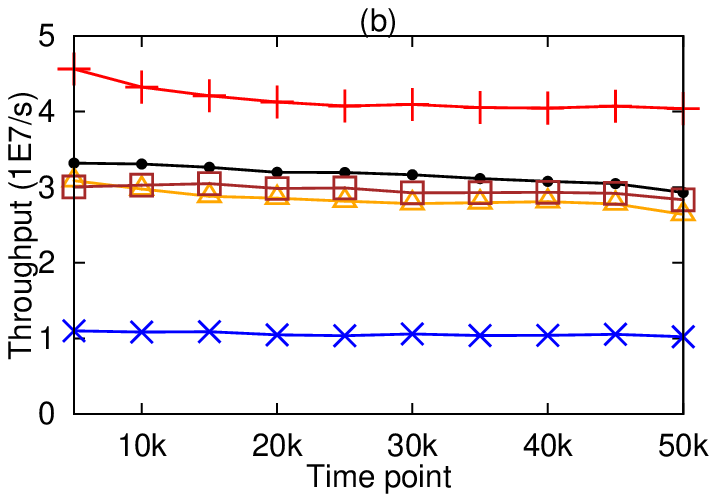}
\end{minipage}
\caption{Online computing, over ILD}
\label{exp:ild-stream}
\end{figure}

\subsubsection{Evaluation on Online Computing}
\label{sect:online-exp}

As long as online labeling is available (discussed in Section \ref{sect-introduction-intuition}), 
the proposed IMR is applicable. 
Remarkably, by interpreting all the historical data as one labeled segment, 
the direct calculation of repairs without iterating in Section \ref{sect:special-one-label} can be applied for efficient online computing. 

Figure \ref{exp:ild-stream} presents the results of online repairing. 
Since errors are randomly introduced in the dataset, to obtain reliable results, we repeat 10 times for each test with random error introducing and report the average.
Troughput of IMR is stable and comparable to others.
IMR again shows the best repair.
Improvement of IMR compared to ARX is not as significant as in other experiments. 
It is not surprising, referring to the similar Equation \ref{equation:ARX} for ARX and Equation \ref{equation:repair-converge-single} for IMR(1).
The advantage of IMR is that no threshold parameter is required for IMR(1) with one labeled segment in Section \ref{sect:special-one-label}, 
while ARX needs to set threshold $\tau$.

\begin{figure}[t]\centering
\begin{minipage}{\expwidths}\centering
\hspace{-0.5em}%
\includegraphics[width=\expwidths]{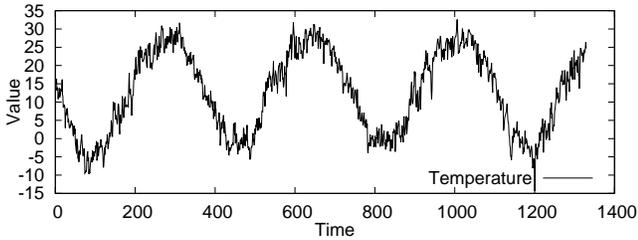}
\end{minipage}
\caption{Cyclic Temperature example}
\label{fig:example-cma}
\end{figure}

\begin{figure}[t]\centering
\begin{minipage}{\expwidths}\centering
\hspace{-0.5em}%
\includegraphics[width=\expwidths]{exp-label-hor} 
\hspace{-0.5em}%
\includegraphics[width=0.5\expwidths]{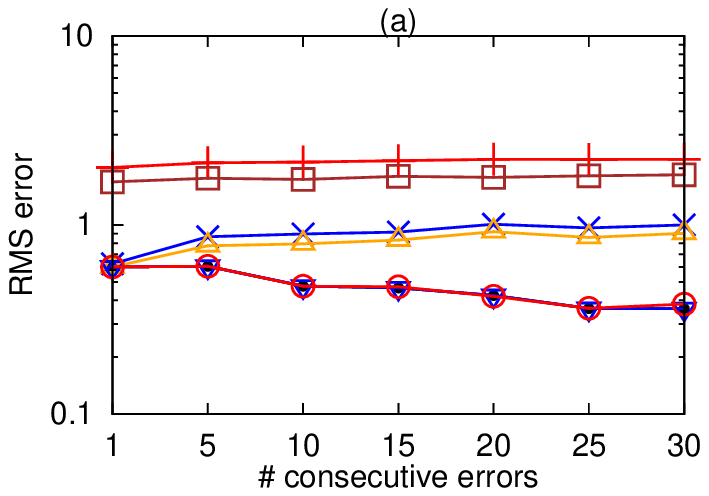}
\hspace{-0.5em}%
\includegraphics[width=0.5\expwidths]{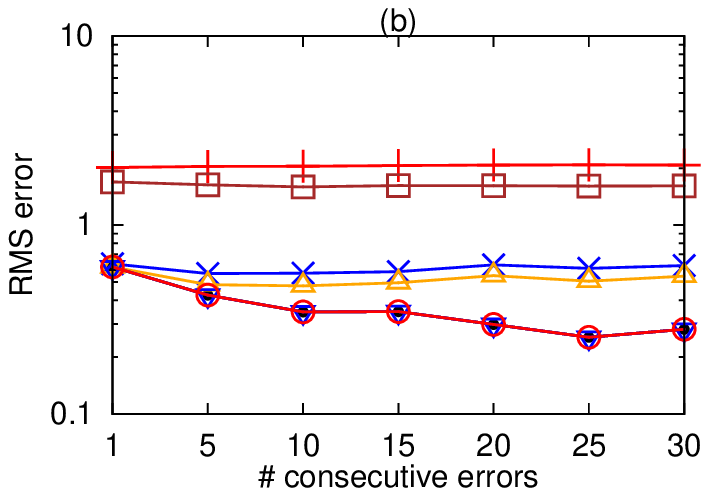}
\end{minipage}
\caption{Varying  number of consecutive errors, under (a)Shift and (b)Innovational error types,  over Temperature with $\tau = 0.2, \mathit{p}=3$ and data size 1.5k}
\label{exp:cma-consize}
\end{figure}

\subsubsection{Experiments on Temperature}
\label{sect:temperature-exp}

To evaluate over cyclic time series, we employ another data set on temperature in years (\url{http://data.cma.cn},
with cyclic patterns as illustrated in Figure \ref{fig:example-cma}). 
Similar to ILD, we inject synthetic errors to the temperature data. 
Figure \ref{exp:cma-consize} presents the results over different error types.
Again, the results are generally similar to those on ILD in Figure \ref{exp:ild-consize}.
That is, the proposed IMR shows significantly better results when dealing with a large number of consecutive errors. 
SCREEN shows no better results in this dataset under Spike errors (i.e., \# consecutive errors = 1), since the clean data also contain a large number of Spikes as illustrated in Figure \ref{fig:example-cma} and cannot be distinguished from errors.

%-------------------------------------------------------------------------
\section{Related Work}\label{sect:related}

The idea of performing repair in multiple iterations has also been studied \cite{DBLP:conf/icde/VolkovsCSM14}, where the past repairs could help in recommend more accurate repairs in the future.
The continuous data cleaning approach \cite{DBLP:conf/icde/VolkovsCSM14} however is not directly applicable in our problem, 
since it employs FD constraints which is not available over time series data.

\subsection{Anomaly Detection over Temporal Data}
\label{sect-related-anomaly}

AR and ARX indeed have been widely used for anomaly detection in various areas such as economics and social surveys \cite{box1994time,brockwell2016introduction}.
We consider ARX \cite{park2005outlier} in this study, since this approach can utilize the labeled truth and could be adapted to cooperate with the minimum change principle in data repairing
(as shown in Section \ref{sect:framework}). 

Hellerstein \cite{hellerstein2008quantitative} surveys methods 
on cleaning errors in quantitative attributes of large databases, 
where time series data are also discussed as a special scenario. 
For instance, 
Tsay \cite{tsay1988outliers} presents unified methods for detecting and handling outliers and structure changes in a univariate timeseries. Iterative procedures consist of specification-estimation-detecting-removing cycles to handle one-by-one the most significant disturbance.
ARIMA \cite{otto1990two,box1994time} is a  general parametric family of time series, consisting of autoregressive process and moving-average process. It can also incorporate a wide range of nonstationary series.

\subsection{Smoothing-based Cleaning}
\label{sect-related-smooth}

Smoothing techniques are often employed to eliminate noisy data.
For example, the simple moving average (SMA) \cite{brillinger2001time} smooths time series data by computing the unweighted mean of the last $k$ points.
Instead of weighting equally, 
the exponentially weighted moving average (EWMA) \cite{gardner2006exponential} assigns exponentially decreasing weights over time.
As indicated in \cite{DBLP:conf/sigmod/SongZWY15}, 
also illustrated in Figure \ref{fig:example-dirty} in Example \ref{fig:example-dirty} 
and observed in Figure \ref{fig:example-gps} (EWMA) in the experiments,
the smoothing methods may seriously alter the original correct data, and thus have low repair accuracy.
In contrast, our minimum change-based approach, applying only high confidence repairs, could preserve most the original values with a better repair accuracy. 

\subsection{Constraint-based Cleaning}
\label{sect-related-constraint}

Constraint-based repairing is widely considered in cleaning dirty data, such that the repaired data satisfies some given constraints and the repair modification is minimized \cite{DBLP:conf/sigmod/BohannonFFR05,DBLP:conf/icde/ChuIP13}.
The constraints are often defined over multiple attributes \cite{DBLP:journals/tods/Song011,DBLP:journals/vldb/Song0Y13,DBLP:journals/tkde/Song0C14}, which are not available in a univariate time series considered in this study. 
To clean sequential data, 
existing study \cite{DBLP:conf/sigmod/SongZWY15} employs a class of speed constraints declaring that the speeds of value changes should be bounded. 
The repairing is thus to modify the sequence towards the satisfaction of such speed constraints.
This constraint-based repairing falls short in two aspects: 
(1) it cannot handle a sequence of continuous errors, and 
(2) the labeled truth is not utilized. 
Rather than hard constraints, the cleaning is further extended to maximizing the likelihood w.r.t.\ speeds \cite{DBLP:conf/sigmod/ZhangSW16}.
Finally, the constraints are also utilized to clean the timestamps \cite{DBLP:journals/pvldb/SongC016} and the qualitative event data \cite{DBLP:journals/tkde/0001SZLS16}, which are different and not applicable to the quantitive time series data studied in this paper.

%=========================================================================
\section{Conclusion}
\label{sect:conclusion}

In this paper, we study the problem of repairing dirty time series data, 
given the labeled truth of some data points. 
(1)
While existing anomaly detection techniques could be adapted to repairing, we argue that 
significant deviation (between observation and predication) based anomaly detection
is inconsistent with the minimum change principle in data repairing.  
Our experiments over real datasets illustrate such inconformity of applying anomaly detection in anomaly repairing. 
(2)
We thereby propose an iterative minimum repairing (IMR) algorithm. 
By creatively performing one minimum repair in each iteration of error predication, 
the algorithm bonds the beauty of capturing temporal nature in anomaly detection 
with the minimum change in data repairing. 
Again, the experiments demonstrate the superiority of our proposal.
Remarkably, in contrast to anomaly detection approaches AR and ARX, our proposed IMR is not sensitive to the setting of order $\mathit{p}$, i.e., a small $\mathit{p}$ is sufficient to achieve high repair accuracy with low time costs. 
(3)
The convergence of IMR is explicitly analyzed. 
In particular, we show that the converged repair result could be directly calculated without iterative computing in certain cases, 
which enables efficient online repairing over streaming data. 
It is worth noting that unlike the existing ARX, no threshold needs to be specified for IMR in online computing.
(4) Finally, we design efficient pruning and incremental computation, 
which reduce the complexity of parameter estimation from linear time to constant time. 
Experiments illustrate the significant improvement on time performance by pruning and incremental computation.

\section*{Acknowledgment}

This work is supported in part by
National Key Research Program of China under Grant 2016YFB1001101;
China NSFC under Grants 61572272, 61325008, 61370055, 61672313 and 61202008;
Tsinghua University Initiative Scientific Research Program.
Shaoxu Song is a corresponding author.

\bibliographystyle{abbrv}
\bibliography{anomaly,bib/stream,bib/quality,bib/dependency,bib/method,bib/theory}

\vspace{-4em}
\includegraphics[height=0in]{tail}
\end{document}